\documentclass[journal, onecolumn, draftcls, 12pt]{IEEEtran}
\usepackage{amsmath}
\usepackage{amsfonts}
\usepackage{amssymb}
\usepackage{amsthm}

\newtheorem{mytheorem}{Theorem}
\usepackage{graphics}
\usepackage{graphicx}

\usepackage{comment}
\usepackage{cite}

\usepackage{newtxmath}

\usepackage{newtxtext}
\usepackage{bm}
\usepackage{enumerate}

\usepackage[bookmarks,colorlinks]{hyperref}
\usepackage{glossaries}
\newacronym{GBC}{GBC}{Gaussian broadcast channel}
\newacronym{DPC}{DPC}{dirty paper coding}
\newacronym{ZF}{ZF}{zero forcing}
\newacronym{MIMO}{MIMO}{multiple-input multiple-output}
\newacronym{SINR}{SINR}{signal-to-interference-plus-noise ratio}
\newacronym{SNR}{SNR}{signal-to-noise ratio}
\newacronym{AWGN}{AWGN}{additive white Gaussian noise}
\newacronym{MMSE}{MMSE}{minimum mean square error}
\newacronym{MAC}{MAC}{multiple access channel}
\newacronym{KKT}{KKT}{Karush-Kuhn-Tucker}
\newacronym{SVD}{SVD}{singular value decomposition}
\newacronym{PDG}{PDG}{primal-dual gradient}

\usepackage{algorithm}
\usepackage{algorithmic}

\author{Xiang Liu, Tianyao Huang, Yimin Liu
    \thanks{X. Liu, T. Huang and Y. Liu are with the Department of Electronic Engineering, Tsinghua University, Beijing, China (e-mail: liuxiang16@mails.tsinghua.edu.cn, \{huangtianyao; yiminliu\}@tsinghua.edu.cn). T. Huang is the corresponding author.}
}

\title{Transmit Design for Joint MIMO Radar and Multiuser Communications with Transmit Covariance Constraint}

\begin{document}
    \maketitle
    
    \begin{abstract}
        In this paper, we consider the design of a multiple-input multiple-output (MIMO) transmitter which simultaneously functions as a MIMO radar and a base station for downlink multiuser communications. In addition to a power constraint, we require the covariance of the transmit waveform be equal to a given optimal covariance for MIMO radar, to guarantee the radar performance. With this constraint, we formulate and solve the signal-to-interference-plus-noise ratio (SINR) balancing problem for multiuser transmit beamforming via convex optimization. Considering that the interference cannot be completely eliminated with this constraint, we introduce dirty paper coding (DPC) to further cancel the interference, and formulate the SINR balancing and sum rate maximization problem in the DPC regime. Although both of the two problems are non-convex, we show that they can be reformulated to convex optimizations via the Lagrange and downlink-uplink duality. In addition, we propose gradient projection based algorithms to solve the equivalent dual problem of SINR balancing, in both transmit beamforming and DPC regimes. The simulation results demonstrate significant performance improvement of DPC over transmit beamforming, and also indicate that the degrees of freedom for the communication transmitter is restricted by the rank of the covariance.

    \end{abstract}
    
    \section{Introduction}

    Joint radar and communications on a single platform is an emerging technique which can reduce the cost of the platform, achieve spectrum sharing, and enhance the performance via the cooperation  of  radar and communications \cite{paul_survey_2017-1,liu_mu-mimo_2018, 9127852}.  
    Because of these promising advantages, numerous schemes  are proposed in recent years to implement joint radar and communications, including multi-functional  waveform design \cite{5776640, zhang_modified_2017, mccormick_simultaneous_2017-1, kumari_ieee_2018, liu_toward_2018, 9420261},  information embedding \cite{blunt_embedding_2010, hassanien_dual-function_2016,7575457, wang_dual-function_2018, 9093221, 9345999}, joint transmit beamforming \cite{mccormick_simultaneous_2017-1, liu_mu-mimo_2018, liu_toward_2018, 9124713, liu2021cramerrao} and so on.  

    We focus on the joint transmit beamforming scheme here, which achieves spatial multiplexing of  radar and communications by forming multiple transmit beams towards the radar targets and communication receivers.
    Previous works based on joint transmit beamforming mainly consider the joint design of a \gls{MIMO} radar and downlink multiuser MIMO communications.
    In particular, these works consider the optimization of the MIMO radar performance, such as the beam pattern mismatch \cite{liu_mu-mimo_2018,9124713} and Cramér-Rao Bound \cite{liu2021cramerrao}, under  individual  \gls{SINR} constraints at the communication receivers.
    Alternatively, some variants of the design \cite{liu_mu-mimo_2018,liu_toward_2018} simultaneous optimize the  performance of radar and communications in the objective function.
    However, MIMO radars exhibit performance trade-off with multiuser communications in these works.
    In other words, to guarantee the  \glspl{SINR} at users, the achievable performance of MIMO radar is worse than the counterpart of a separate MIMO radar without considering communications.
    In high speed communication scenario, the performance loss of MIMO radar can be significant to achieve high \glspl{SINR} at  users \cite{9124713}.

    In our work, we consider a joint MIMO radar and multiuser communication system, in which radar is the primary function and communication is the secondary function. 
    Under this scenario, the efficiency of MIMO radar should be first guaranteed without any performance loss of radar. In this regard, we study the joint design, where we optimize the communication performance under the requirement that the radar maintains its optimal performance without communications.
    Literature on MIMO radar reported that the performance of MIMO radars highly depends on the  covariance of the transmit waveform \cite{stoica_probing_2007,fuhrmann_transmit_2008, li_range_2008,1703855}.
    Therefore, we formulate transmitter optimizations for communications, under  the transmit covariance constraint that the covariance of the transmit waveform is equal to the given optimal  one for MIMO radar without communication function.
    The proposed approach in \cite{liu_toward_2018} considers a similar constraint, but constrains the instantaneous covariance and needs to optimize the instantaneous transmit waveform with the constraint.
    Different from \cite{liu_toward_2018}, we constrain the average covariance, and optimize the precoding matrices as in \cite{liu_mu-mimo_2018,9124713,liu2021cramerrao}.

    At the transmitter,  linear precoding technique is usually applied to generate the transmit waveform, which performs transmit beamforming to improve the \glspl{SINR} at downlink users \cite{730452,visotsky_optimum_1999,wiesel_zero-forcing_2008,bjornson_optimal_2014}.
    For transmit beamforming, we formulate the \gls{SINR} balancing \cite{wiesel_linear_2006,9027103} problem for multiuser communications, which designs the precoding matrices by maximizing the worst \gls{SINR} at the users with the transmit covariance constraint.
    We show that the  problem can be reformulated to a linear conic optimization \cite{Luenberger2016}, and further propose an  iteration method to solve its Lagrange dual \cite{boyd_vandenberghe_2004}, which has a low complexity and converges fast.

    Despite the low complexity of transmit beamforming, the numerical results show that, the transmit covariance constraint, introduced by the radar function, typically results in low \glspl{SINR}  via  transmit beamforming.
    This still happens even if the \gls{SNR} is high, because the inter-user interference cannot be canceled under such constraint.
    To further eliminate the interference, we investigate the application of  \gls{DPC}  \cite{costa_writing_1983-1}, which reveals that the interference in an \gls{AWGN} channel does not reduce the capacity if the interference is known at the transmitter, and was  applied to the interference canceling in downlink multiuser communications  \cite{1207369,viswanath_sum_2003, vishwanath_duality_2003, 1327794, weingarten_capacity_2006, lee_dirty_2006, 4063519, 4203115}.

    We apply  \gls{DPC}  for the transmit design of joint radar and communications, and formulate the \gls{SINR} balancing problem for \gls{DPC} with the strict radar performance constraint.
    Considering the  optimization is non-convex,  we derive its equivalent dual problem from the Lagrange dual of the power minimization problem, which finds the minimal transmit power to achieved given \glspl{SINR} at users.
    The dual problem has a convex structure, and we proposed a gradient based iteration method to solve it.
    Meanwhile, we consider to maximize the sum rate of the users, which is still a non-convex optimization problem.
    Using the downlink-uplink duality, we show that it is equivalent to the sum rate maximization for an equivalent uplink channel, which is expressed as a convex-concave saddle point problem, and further prove that the saddle point can be obtained by solving an equivalent linear conic optimization.	
    The simulation results in the \gls{DPC} regime show that the \gls{DPC} can significantly improve the obtained \glspl{SINR} at users compared to transmit beamforming. 
    
    While the proposed \gls{DPC} approaches relieve the interference issues for downlink users, it should be noted that the hard constraint on transmit covariance matrix essentially limits the communication performance. In particular, we reveal that the degrees of freedom of the communication transmitter is limited by the rank of the transmit covariance, from the following two observations, with the number of users denoted by $K$ and the rank denoted by $r_o$: 
    \begin{itemize}
        \item The balanced \gls{SINR} for \gls{DPC} encounters a significant decrease when $K$ exceeds $r_o$;
        \item Under a high transmit \gls{SNR}, the maximal sum rate for \gls{DPC} is asymptotically affine in the transmit \gls{SNR} in dB.
        The multiplexing  gain \cite{978730}, (i.e.,  the rate gain in bits/channeluse for  every $3$dB  power  gain), is $K$ with a  power constraint \cite{lee_dirty_2006}, while it reduces to $\min(K, r_o)$ with the transmit covariance constraint.      
    \end{itemize}

    The rest of the paper is organized as follows. In Sec.~\ref{sec-2}, we give the signal model, introduce the transmit covariance constraint, and formulate the optimization for communications, via both transmit beamforming and \gls{DPC}.
    In Sec.~\ref{sec-3}-\ref{sec-5}, we study the numerical methods to solve the \gls{SINR} balancing  for transmit beamforming,  \gls{SINR} balancing for \gls{DPC} and sum rate maximization for \gls{DPC}, respectively.
    We demonstrate the communication performance and the convergence property  of the proposed iteration algorithms via numerically results  in Sec.~\ref{sec-6}.
    Sec.~\ref{sec-7} draws the conclusion.
    
    \subsubsection*{Notations} 
    For a matrix $\bm X$, we denoted its $(i,j)$-th element by $\bm X_{i,j} $ or $[\bm X]_{i,j}$.
    For an integer $K > 0$,  $\bm 1_K$ represents a $K$-dimensional vector whose elements are all $1$.
    In this paper, $\mathcal{C} (\cdot)$ represents the column space of a matrix, and $(\cdot) ^ \dagger$ represents the Moore-Penrose inverse \cite{8a8342c0aefb4c0b9885d8e81b6fa219}.
    
    \section{Joint transmit design problems} \label{sec-2}
    
    Consider a joint transmitter which simultaneously functions as a \gls{MIMO} radar transmitter and a base station for downlink multiuser communications.
    In the transmitter, radar and communications share transmit signal, whose expression is given in Sec.~\ref{sec-2-1}, following \cite{9124713}. 
    Considering the radar performance, we introduce a transmit covariance constraint to the transmit signal in Sec.~\ref{sec-2-2}.
    With this radar constraint, we formulate the general transmit beamforming optimizations for communications in Sec.~\ref{bf}, and  extend the optimizations  to the \gls{DPC} regime in Sec.~\ref{dpc}.

    \subsection{Shared transmit signal} \label{sec-2-1}

    The transmitter is equipped with a transmit array with $M$ antennas and sends independent communication symbols to $K $ users, where $K \leq M$.   
    The average transmit power is $P$.   
    The transmit signal $\bm x(n) $ for the shared transmit array is generated by the joint linear precoding scheme in \cite{9124713}.
    In particular,  $\bm x(n) $  is the sum of linear precoded radar waveforms and communication symbols, given by
    \begin{equation} \label{eq-xn}
        \bm{x}(n) = \bm{W}_r \bm{s} (n) +  \bm{W}_c \bm{c} (n), \ n = 0,\ldots,N-1,
    \end{equation}   
    where  $N$ is the number of samples.
    Here, $\bm s(n)  = [ s_1(n), \ldots, s_M(n) ] ^ T$ includes $M$ orthogonal radar waveforms, and the $ M \times M $ matrix $\bm W_r$ is the precoding matrix for radar \cite{hassanien_dual-function_2016}.
    The orthogonality of radar waveforms means that $   (1/N) \sum_{n=0}^{N-1} \bm s(n) \bm s ^ H (n) = \bm I _ M$.
    The $K$ parallel communication symbols to the users are contained in $ \bm c(n) = \left[ c_1(n), \ldots, c_K(n)  \right] ^ T$, precoded by the $M \times K$ matrix $\bm W_c$.

    Following \cite{bjornson_optimal_2014,liu_mu-mimo_2018,9124713}, we rely on the following conditions to the communications symbols and radar waveforms:
    \begin{enumerate}[(a)]
        \item The communication symbols to different users are mutually independent, have zero mean,  and are normalized to have unit average power. Therefore,  $\mathbf{E}(\bm s(n) \bm s^H(n)) = \bm I_M $.
        \item The radar waveforms and communication symbols are statistically independent.
    \end{enumerate}

    Given $ \{ \bm s(n) \}$ and  $ \{ \bm c(n) \}$, transmit design for joint MIMO radar and communication becomes designing $\bm W_r$ and $\bm W_c$.

    \subsection{Transmit covariance constraint for radar} \label{sec-2-2} 
    
    The radar is monostatic so that the communication signals  can also be used for target detection because they are completely known at the radar receiver.
    Unlike phased array radars, \gls{MIMO} radars transmit independent or partially correlated signals from the array elements.
    The performance of \gls{MIMO} radar highly depends on its  transmit covariance
    \begin{equation} \label{eq-R}
        \bm R = \mathbf{E} \Big\{ \frac{1}{N} \sum_{n=0}^{N-1} \bm x (n) \bm x ^ H (n)  \Big\}.
    \end{equation}   
    It was shown that the transmit beam pattern \cite{stoica_probing_2007}, the angular estimation accuracy \cite{1703855, li_range_2008}  and the detection performance of radars \cite{1703855} is determined by $\bm R$.    
    Substituting \eqref{eq-xn} into \eqref{eq-R}, $\bm R$ is  given by
    \begin{equation} \label{eq-pt-ro}
        \bm R = \bm W_r \bm W_r ^ H + \bm W_c \bm W_c ^ H.
    \end{equation}
    Given the average transmit power $P$, $\bm R$ should obey $\mathrm{tr}(\bm R) = P$. 
    To guarantee the radar performance, in a solely \gls{MIMO} radar without communications,  the  transmit covariance is optimized, yielding $\bm R_{o}$, under a power constraint as in \cite{stoica_probing_2007,li_range_2008,fuhrmann_transmit_2008}. 
    
    Then in the joint design considered in this paper,  $\bm W_c$ and $\bm W_r$ are constrained so that the obtained $\bm R$ in \eqref{eq-pt-ro} equals to $\bm R_{o}$. Hence, the joint radar and communications system achieves the optimal radar performance as the  solely radar system without communications.  
    In the following, we design $\bm W_c$ and $\bm W_r$ by transmit beamforming and \gls{DPC}, respectively, to optimize the communication performance under this constraint. 
    
    This approach is  different from existing precoder design methods in \cite{liu_mu-mimo_2018, 9124713, liu2021cramerrao} where they sacrifice the radar performance to achieve the desired \glspl{SINR} for communications.  Particularly in their methods, $\bm W_c$ and $\bm W_r$ are constrained to meet the minimum requirements on  \glspl{SINR}, and are optimized to improve the radar performance. 
    
    
    \subsection{Transmit beamforming for multiuser communications} 
    \label{bf}
    
    For downlink multiuser communication, transmit beamforming is performed to increase the signal power at intended users and reduce interference to non-intended users \cite{730452, bjornson_optimal_2014}. 
    Here, a vector \gls{GBC} \cite{1054727} is considered in which each user is equipped with a single receive antenna.
    The channel is denoted  by a $K \times M$ matrix $\bm H$.
    The channel output of the   \gls{GBC}  is given by \cite{9124713}
    \begin{equation} \label{eq-yn}
        \bm   r (n) = \bm H \bm x(n)  + \bm  v(n)  = 
        \bm H \bm W_c \bm c (n) + \bm H  \bm W_r \bm s(n) + \bm v (n).
    \end{equation}
    Here, the $k$-th elements of $\bm r(n)$ represents the received signal at the $k$-th user,  and $\bm v(n) $ is complex \gls{AWGN} whose covariance is $\sigma ^ 2 \bm I_K$. For convenience, we let $\sigma^2 = 1$ in the sequel.
    
    In \eqref{eq-yn}, each user receives the mixture of its own signals, the interference from other users, the radar signal and the noise.  
    Let $\bm F = \bm H \bm W_c$ and $  \bm G = \bm H \bm W_r$.
    For the $k$-th user,   the sum power of received signal, including both desired signal and interference is $ \sum_{i = 1}^{M}  |  \bm{F}_{k,i} |^2 + \sum_{i = 1}^{M}  |  \bm{G}_{k,i} |  ^2 =  [ \bm H \bm R_o \bm H ^ H ]_{k,k}$,    
    and the power of desired signal  is $| \bm F_{k,k} | ^ 2$.
    Then,  in the transmit beamforming regime,  the \gls{SINR} at the $k$-th user is given by \cite{liu_mu-mimo_2018,9124713}
    \begin{equation} \label{eq-sinr}
        \mathrm{SINR}_k = \frac{| \bm F_{k,k} | ^ 2 }{\sum_{i \neq k}  |  \bm{F}_{k,i} |^2 + \sum_{i = 1}^{M}  |  \bm{G}_{k,i} |  ^2 + 1 },
    \end{equation}
    for $k = 1,\ldots,K$. 
    In the regime of \gls{DPC}, the interference is treated differently. Hence, the definition of \gls{SINR} is different, as will be introduced later in~\ref{dpc}.
    
    Our goal is to maximize a  utility function $f(\mathrm{SINR}_1,\ldots,\mathrm{SINR}_K)$ which is increasing in the \glspl{SINR} \cite{bjornson_optimal_2014}, with the transmit covariance constraint.
    The optimization problem is stated as
    \begin{subequations} \label{eq-max-bf}
        \begin{align}
            \max_{\bm W_c, \bm W_r}  \ & \ f(\mathrm{SINR}_1,\ldots,\mathrm{SINR}_K) \\
            \mathrm{s.t.} \ & \ \bm R_o = \bm W_r \bm W_r ^ H + \bm W_c \bm W_c ^ H. \label{eq-max-bf-b}
        \end{align}
    \end{subequations}
    There are some common utility functions in existing works.
    For \gls{SINR} balancing, the utility functions is the worst \gls{SINR} of the users, given by \cite{wiesel_zero-forcing_2008}
    \begin{equation} \label{eq-f1}
        f(\mathrm{SINR}_1,\ldots,\mathrm{SINR}_K) = \min_{1 \leq k \leq K} \ \mathrm{SINR}_k.
    \end{equation}
    For  sum rate maximization, the  utility function is \cite{wiesel_zero-forcing_2008, bjornson_optimal_2014}
    \begin{equation} \label{eq-f2}
        f(\mathrm{SINR}_1,\ldots,\mathrm{SINR}_K) = \sum_{k=1}^{K} \log(1+\mathrm{SINR}_k).
    \end{equation}

    We then reformulate the optimization problems with respect to $\bm F$. 
    Letting $\bm R_h = \bm H \bm R_{o} \bm H ^ H$,  it can be shown that the constraint in \eqref{eq-max-bf-b} is equivalent to \cite{schur}
    \begin{equation} \label{eq-ff}
        \bm F \bm F^H \preceq \bm R_h \  \Leftrightarrow \
        \left[
        \begin{array}{cc}
            \bm R_h 	&  \bm F \\
            \bm F^H 	&   \bm I_K 
        \end{array}   	
        \right] \succeq 0,
    \end{equation}
    which is convex.
    Meanwhile, we let  $s_k = ([\bm R_h]_{k,k} + 1 )^{1/2}$, simplifying the \gls{SINR} at the $k$-th as
    \begin{equation}
        \mathrm{SINR}_k = \frac{| \bm F_{k,k} | ^ 2 }{s_k ^ 2 - | \bm F_{k,k} | ^ 2  }, \ k = 1,\ldots,K.
    \end{equation}
    Introducing the slack variables $\gamma_k \leq \mathrm{SINR}_k $, for $k = 1,\ldots, K$, we reformulate \eqref{eq-max-bf} into an optimization with respective to $\bm F$ and $\bm \gamma =  [ \gamma_1, \ldots, \gamma_K]^T$:
    \begin{equation} \label{eq-max-bf1}
         \max _{\bm F,\bm \gamma} \   f(\gamma_1,\ldots, \gamma_K) , \ \ 
            \mathrm{s.t.}	  \   \eqref{eq-ff} \ \mathrm{and} \  | \bm F_{k,k} |  \geq \sqrt{\frac{\gamma_k}{1+\gamma_k}} s_k, \ k = 1,\ldots,K.
    \end{equation}
    After solving \eqref{eq-max-bf1}, the optimum of the original problem in \eqref{eq-max-bf} can be computed by
    \begin{equation}
        \bm W_c = \bm R_{o}^{1/2}  ( \bm H \bm R_ o ^ {1/2} )^{\dagger} \bm F, \  \bm W_r = ( \bm  R _ o - \bm W_c \bm W_c ^ H   )^{1/2} .
    \end{equation}  
    For \gls{SINR} balancing, the solver for \eqref{eq-max-bf1} will be provided in Sec.~\ref{sec-3}.

    \subsection{Dirty paper coding for multiuser communications} \label{dpc}

    The numerical results in  Sec. \ref{sec-simu2} and in \cite{9124713}  indicate that  the achievable  \glspl{SINR} with \eqref{eq-max-bf} and  \eqref{eq-max-bf1} may be  low  with the transmit covariance constraint, and some trade-off  designs \cite{liu_mu-mimo_2018,liu_toward_2018,9124713} were proposed to improve the  \gls{SINR} by relaxing the constraint.
    To further improve the \glspl{SINR}, one can perform non-linear precoding techniques, which eliminate the interference by encoding the communication signals to adapt the interference.
    In particular,  we consider \gls{DPC}  \cite{costa_writing_1983-1}, which reveals that the capacity of an \gls{AWGN} channel corrupted by interference equals to the capacity of  an  interference-free \gls{AWGN} channel if the interference is known at the transmitter.
    \Gls{DPC} was applied to \gls{GBC} to eliminate the effect of inter-user interference \cite{1207369, viswanath_sum_2003, vishwanath_duality_2003, 1327794}, and was shown to be able to achieved the capacity region of \gls{MIMO} \gls{GBC} \cite{weingarten_capacity_2006}.

    We apply \gls{DPC} to the \gls{GBC} in \eqref{eq-yn} by serially encoding the source signal of each user.
    The encoding operations are conducted in the order $\{1,\ldots,K\}$.
    When performing \gls{DPC} for the $k$-th user, $\{c_1(n)\},\ldots,\{ c_{k-1}(n)\}$ are already encoded while $\{  c_{k+1}(n) \},\ldots,\{  c_{K}(n) \}$ are not encoded yet.
    Thus, the interference from the $1,\ldots,(k-1)$-th user is known while the interference from the  $k+1,\ldots,K$-th user is unknown at the transmitter.
    The radar interference is also known at the transmitter.
    Therefore, the effective \gls{SINR} at the $k$-th user in the \gls{DPC} regime is \cite{viswanath_sum_2003}
    \begin{equation} \label{eq-sinr-dpc}
        \mathrm{SINR}_k^{\mathrm{dpc}} = \frac{|\bm F_{k,k} | ^ 2 }{\sum_{i > k}  | {\bm F}_{k,i} |^2 +  1 },
    \end{equation}
    for $k = 1,\ldots,K$.     
   
    Comparing \eqref{eq-sinr} and \eqref{eq-sinr-dpc}, it is observed that \gls{DPC} improves the \gls{SINR} compared with transmit beamforming by eliminating the    interference.
    It is worth noting that when $\bm R_h$ is non-singular, one can perform \gls{ZF} \gls{DPC}, namely completely cancel the interference, by computing a lower triangular 
    $\bm F$ via the Cholesky decomposition \cite{8a8342c0aefb4c0b9885d8e81b6fa219} of $\bm R_h$, while \gls{ZF} transmit beamforming is generally not applicable, since it requires $\bm R_h $ be a diagonal matrix \cite{9124713}.

    Similarly, we maximize the utility function in the \gls{DPC} regime.
    The optimization problem is stated as
    \begin{equation} \label{eq-max-dpc}
        \max_{\bm W_c, \bm W_r} \ \  f(\mathrm{SINR}_1^{\mathrm{dpc}} ,\ldots,\mathrm{SINR}_K^{\mathrm{dpc}} ) , \ 
        \mathrm{s.t.} \  \eqref{eq-max-bf-b} .
    \end{equation}
    We note that the  \glspl{SINR}  in the \gls{DPC} regime  also explicitly depend on   $\bm F$.
    As in the transmit beamforming regime, we can reformulate  \eqref{eq-max-dpc} into  an optimization with respect to $\bm F$.
    Introducing the slack variables $\gamma_k \leq \mathrm{SINR}_k  ^ {\mathrm{dpc}}$, for $k = 1,\ldots, K$,  \eqref{eq-max-dpc}  is equivalent to
    \begin{equation}  \label{eq-max-dpc1}
        \max _{\bm F,\bm \gamma} \  f(\gamma_1,\ldots, \gamma_K), \ \	\mathrm{s.t.}  \ \eqref{eq-ff} \ \mathrm{and}    \
           \frac{1}{\gamma_k}  |\bm F_{k,k} | ^ 2 \geq  \sum_{i > k}  | {\bm F}_{k,i} |^2 +  1 , \ k = 1,\ldots,K.  
    \end{equation}
    
    In Sec.~\ref{sec-4} and~\ref{sec-5}, we will consider the two criteria, \gls{SINR} balancing and sum rate maximization, for \gls{DPC}, repectively. 
    Prior to these \gls{DPC} approaches, we first discuss below in Sec.~\ref{sec-3} the \gls{SINR} balancing in the transmit beamforming  regime as a benchmark. 
    
    \section{SINR balancing for transmit beamforming} \label{sec-3}
    
    In this section, we provide optimization methods to solve the SINR balancing problem in the transmit beamforming regime.
    The optimization problem is \eqref{eq-max-bf1} with the target function given by 
    $  f(\gamma_1,\ldots, \gamma_K) = \min _{1 \leq k \leq K} \, \gamma_k$. 
    We will first reformulate it into a linear conic optimization  that can be effectively solved by on-the-shelf optimization solvers, shown in Sec.~\ref{subsec: linear conic}. 
    Later in Sec.~\ref{sec-bf-2}, we will further provide a more efficient, dual program based method. 
    
    \subsection{Conic optimization solution}
    \label{subsec: linear conic}
    
    While the \gls{SINR} constraint in \eqref{eq-max-bf1} are non-convex, we can convert it to a convex one, and hence reformulate \eqref{eq-max-bf1} to a linear conic optimization.  
    
    Defining the balanced \gls{SINR}  $\gamma =  \min _{1 \leq k \leq K} \, \gamma_k$, we reformulate the optimization into
    \begin{equation}  \label{eq-sinr-p1}
        \max _{\bm F,\gamma} \ \gamma, \ \
            \mathrm{s.t.}	  \  \eqref{eq-ff} \ \mathrm{and} \   | \bm F_{k,k} |  \geq \sqrt{\frac{\gamma}{1+\gamma}} s_k, \ k = 1,\ldots,K.
    \end{equation}
    %
    Note that for a feasible $\bm F$ to \eqref{eq-sinr-p1}, rotating its $k$-th column by a scalar phase factor $e^{j \theta_k}$ does not violate the feasibility \cite{4203115, bjornson_optimal_2014}.
    Therefore, we only need to consider $\bm F$ with real diagonal elements.
    Introducing a new variable $t = \sqrt{\gamma / (1+\gamma)}$, \eqref{eq-sinr-p1} is equivalent to the following:
    \begin{equation} \label{eq-sinr-p2}
         \max _{\bm F,t} \  t, \ \
            \mathrm{s.t.}	  \  \eqref{eq-ff} \ \mathrm{and} \   \Re \{ \bm F_{k,k} \}  \geq t s_k, \ k = 1,\ldots,K.
    \end{equation}
    This is now solvable by linear conic programming. 
    
    The scale of \eqref{eq-sinr-p2} can be further reduced when $\bm R _h$ is singular.  
    We let $ r \leq K$ be the rank of $\bm R_h$, and write the  eigen decomposition of $\bm R_h$ as $ \bm R_h = \bm U \bm \Sigma_r \bm U^H $,  where $\bm \Sigma_r $ is a $r \times r$ diagonal matrix and $\bm U ^ H \bm U  = \bm I_r$.
    Following the constraint in  \eqref{eq-ff}, we let $\bm F =  \bm U \bm \Sigma_r ^ {1/2} \bm F_u$. Then  \eqref{eq-sinr-p2} becomes an optimization with respect to $\bm F_u$:
    \begin{equation} \label{eq-sinr-p3}
         \max _{\bm F_u,t} \  t,  \ \
            \mathrm{s.t.}	  \  \bm F_u \bm F_u ^ H \preceq \bm I_r , \ \Re \{ \bm u_k ^ H \bm f_{k} \}  \geq t s_k, \ k = 1,\ldots,K, 
    \end{equation}
    where $\bm f_{k}$ is the $k$-th column in $\bm F_u$ and  $ \bm u_k ^ H$ is the $k$-th row in $\bm U \bm \Sigma_r ^ {1/2} $.
    In \eqref{eq-sinr-p3}, the dimension of $\bm F_u$ reduces to $r \times K$ if $r < K$.

    \subsection{Dual program solution} \label{sec-bf-2}
    
    Here, we propose a simpler iteration method for \eqref{eq-sinr-p3} based on its Lagrange dual program. 
    Define the Lagrange function \cite{boyd_vandenberghe_2004} of \eqref{eq-sinr-p3} by
    \begin{equation}
    	\mathcal{L}_1 (\bm F_u,t, \bm Y, \bm d ) = t + \mathrm{tr} \left\{  \bm Y ( \bm I_r -  \bm F_u \bm F_u^H) \right \} + \sum_{k=1}^K d_k \left( \Re \{ \bm u_k ^ H \bm f_{k} \}  - t  s _ k \right),
    \end{equation}
    where $\bm Y \succeq 0$ and $\bm d = [d_1,\ldots,d_K]^T \geq 0 $ are associated dual variables.
    We first derive the  dual of \eqref{eq-sinr-p3} by maximizing the  Lagrange function  with respect to $\bm F_u$ and $ t $, and then show that the  dual problem is equivalent to  
    \begin{equation} \label{eq-dual-problem-d}
        \min _ {  \bm d \geq 0 } \mathrm{tr} \Big\{ (\bm D \bm D ^ H    )^{1/2}   \Big\}, \ \mathrm{s.t.}  \ 		\bm s^ T \bm d =  1,
    \end{equation}
    where $\bm D = \left[ d_1 \bm u_1, \ldots, d_K \bm u_K \right]$ and $\bm s =  [s_1, \ldots, s_K] ^ T$.
    The derivations are given in Appendix~\ref{A-5}.

    Generally, solving \eqref{eq-dual-problem-d} can be  simpler than \eqref{eq-sinr-p3}   as the variable is only a $K$ dimensional vector.
    The  constraints in  \eqref{eq-dual-problem-d}  are all linear, and the objective function is actually the nuclear norm \cite{5452187} of $\bm D$, which is convex in $\bm d$.
    Nuclear norm minimization with linear constraints can be solved  by  optimization softwares such as CVX \cite{cvx, gb08}.
    
    We also propose a gradient projection  based method to solve  \eqref{eq-dual-problem-d}. 
    The key step is to calculate the descend direction at a point $\bm d$, denoted by $\bm  \Delta _d (\bm d)$, under the constraint. 
    Let $h(\bm d)$ denote the  objective function in \eqref{eq-dual-problem-d}, which  is differential if $\bm d$ is  strictly feasible, i.e. $\bm d > 0$ and $\bm s ^ T \bm d = 1$.
    The gradient of $h(\bm d)$, denoted by $\nabla h (\bm d) \in \mathbb{R} ^ K$, is given by
    \begin{equation}
        [ \nabla h (\bm d)]_k  = 
        d_k ^ {-1}  [ ( \bm D^H  \bm D  )^{1/2} ]_{k,k},
    \end{equation}
    for $k = 1,\ldots,K$.   
    Considering the constraints, we compute $\bm  \Delta _d (\bm d)$ from a projection operation \cite{doi:10.1137/S1052623497330963}: 
    \begin{equation}
        \bm  \Delta _d (\bm d) = \mathcal{P}_{\Omega_1} \left( \bm d - \nabla h (\bm d)    \right) - \bm d,
    \end{equation}
    where  $ \mathcal{P}_{\Omega_1} (\cdot)$ is the orthogonal projection onto the constraint set $ \Omega_1 = \{  \bm d \, | \,  \bm d \geq 0, \ \bm s^T \bm d = 1   \}$.
    The projection does not have a close form expression.
    Nevertheless, it can be computed by at most $K$ loops.
    The details on the computation of the projection is omitted here, and is given in Appendix~\ref{A-2}. 
    
    With the obtained $\bm  \Delta _d (\bm d)$ from above loops, 
    gradient projection is performed via the following iterations:
    \begin{equation} \label{eq-gradient-descend}
        \bm d ^ {(\ell+1)} := \bm d ^ {(\ell )} + \alpha ^ {(\ell )} \bm   \Delta _ d (\bm d  ^ {(\ell )} ) , \ \ell =0,1,\ldots 
    \end{equation}
    Here, the step size $\alpha ^ {(\ell )} \in (0,1) $ can be determined by backtracking search \cite{boyd_vandenberghe_2004}.
    The initial value $\bm d ^ {(0)}$ should be strictly feasible. 
    The iterations in \eqref{eq-gradient-descend} can be stopped if $ \|  \bm   \Delta _ d ( \bm d  ^ {(\ell )} )  \|_2  $ is small enough.

    Then we compute $t$ and $\bm F_u$ for the primal problem.
    Let $L$ be the number of performed iterations at convergence.
    With strong duality, $t$ is given by $t := h(\bm d^{(L)})$.
    From the \gls{KKT} conditions \cite{boyd_vandenberghe_2004}, it holds that $\bm Y \bm F_u = \frac{1}{2} \bm D$ at the optimum of \eqref{eq-sinr-p3} and its dual, where $ \bm Y = \frac{1}{2} (   \bm D \bm D ^ H  )^{1/2}  $ according to Appendix~\ref{A-5}.
    Thus, we compute $\bm F_u$ via 
    \begin{equation}
        \bm F_u :=   \big[ \bm D ^{(L)}  ( \bm D ^{(L)}  ) ^ H \big] ^ {-1/2} \bm  D ^{(L)} , 
    \end{equation}
    where $\bm D ^{(L)} = \big[ d_1 ^{(L)} \bm u_1, \ldots, d_K ^{(L)} \bm u_K \big]$.

    %

    \section{SINR balancing for DPC} \label{sec-4}
    
    In this section, we provide optimization methods to solve the \gls{SINR} balancing problem in the \gls{DPC} regime. 
    Like \eqref{eq-sinr-p1}, we define the balanced \gls{SINR}  $\gamma =  \min _{1 \leq k \leq K} \, \gamma_k$, and then the optimization  is expressed as a problem with respect to $\bm F$ and $\gamma$:     
    \begin{subequations} \label{eq-balancing-dpc}
        \begin{align}
            \max_{\bm F, \gamma} \ & \  \gamma, \  \ \mathrm{s.t.} \ \bm F \bm F^H \preceq \bm R_h,  \\
            & \  \frac{1}{\gamma}  |\bm F_{k,k} | ^ 2 \geq  \sum_{i > k}  | {\bm F}_{k,i} |^2 + 1 , \ k = 1,\ldots,K. \label{eq-balancing-dpc-b}
        \end{align}
    \end{subequations}
    
    In Sec.~\ref{subsec: power minimization}, we provide a power minimization based solution, which is solvable with some on-the-shelf optimization toolboxes. Later in Sec.~\ref{sec-dpc-b}, we give a more efficient solution based on its dual problem.

    \subsection{Solution via power minimization}
    \label{subsec: power minimization}
    
    Unlike the \gls{SINR} balancing for transmit beamforming,  \eqref{eq-balancing-dpc} is nonconvex since the constraint in \eqref{eq-balancing-dpc-b} is nonconvex.
    Despite its non-convexity,  \eqref{eq-balancing-dpc} can be solved with a  polynomial time complexity. Firstly, we formulate the corresponding  power minimization problem \cite{visotsky_optimum_1999,wiesel_linear_2006,4203115}, which is a linear conic optimization, solvable with a  polynomial time complexity. Based on the results, the optimizer of the original problem is then obtained with bisection search, also solvable with a  polynomial time complexity. 
    
    To formulate the power minimization problem, we assume that $ \gamma $ is given. 
    We then  
    adjust the transmit power by  a scaling factor $\lambda \geq 0$, and thus the transmit power and covariance becomes $\lambda P$ and $\lambda \bm R_o$, respectively.
    Regarding $\lambda$ as a variable, the minimal transmit power problem seeks for the minimal $\lambda$ under the \gls{SINR} constraints, 
    given by
    \begin{subequations} \label{eq-minimal-power}
        \begin{align}
            \min_{\lambda, \bm F} \ & \  \lambda  , \  \ \mathrm{s.t.} \ \bm F \bm F^H \preceq \lambda \bm R_h, \label{eq-minimal-power-a}  \\
            & \  \frac{1}{\gamma}  |\bm F_{k,k} | ^ 2 \geq  \sum_{i > k}  | {\bm F}_{k,i} |^2 +  1 , \ k = 1,\ldots,K. \label{eq-minimal-power-b}
        \end{align}
    \end{subequations}
    Let $\lambda^*(  \gamma)$ be the optimal value of \eqref{eq-minimal-power}.

    The minimal transmit power problem is a solvable linear conic optimization.
    In \eqref{eq-minimal-power}, the constraint in \eqref{eq-minimal-power-a}   can be recast to a semi-definite constraint as in \eqref{eq-ff}.
    Similar to \eqref{eq-sinr-p2}, we replace $|\bm F_{k,k} |$ by $\Re \{ \bm F_{k,k} \}$ in \eqref{eq-minimal-power-b}, and the constraint becomes  second order cone constraints.
    Therefore, the optimization in \eqref{eq-minimal-power} can be effectively solved by linear conic solvers with a polynomial time complexity.

    The relationship between the original problem  \eqref{eq-balancing-dpc} and \eqref{eq-minimal-power} is built from this observation:  The \gls{SINR} $ \gamma$ is achievable  if and only if $ \lambda ^*(  \gamma) \leq 1 $, namely the minimal transmit power to achieve the \gls{SINR} $\gamma$ is less than $P$.    
    Note that 
    $\lambda^*(\gamma )$ increases monotonically in $\gamma$. 
    Therefore, the optimal $\gamma$ in \eqref{eq-balancing-dpc}, denoted by $\gamma_o^*$, should satisfy $ \lambda^*(\gamma_o^* ) = 1 $, and can be found by a bisection search \cite{wiesel_linear_2006}, which can be finished in polynomial time.
    
    Similar to \eqref{eq-sinr-p3}, \eqref{eq-minimal-power} can be reformulated to an optimization with respect to $\bm F_u$. 
    Let $ r = \mathrm{rank} (\bm R_h)$, $ \bm R_h = \bm U \bm \Sigma_r \bm U^H $ be the eigen decomposition of $\bm R_h$, and $ \bm F = \bm U \bm \Sigma_r ^ {1/2} \bm F_u$.
    Then \eqref{eq-minimal-power} is reformulated to
    \begin{subequations} \label{eq-minimal-power2}
        \begin{align}
            \min_{\lambda, \bm F_u} \ & \ \lambda  , \  \ \mathrm{s.t.} \ \bm F_u \bm F_u^H \preceq \lambda \bm I_r, \label{eq-minimal-power2-a}  \\
            & \  \frac{1}{\gamma}  |\bm u_k ^ H \bm f_k | ^ 2 \geq  \sum_{i > k}  | \bm u_k ^ H \bm f_i |^2 +  1 , \ k = 1,\ldots,K, \label{eq-minimal-power2-b}
        \end{align}
    \end{subequations}
    where $\bm f_{k}$ is the $k$-th column in $\bm F_u$ and  $ \bm u_k ^ H$ is the $k$-th row in $\bm U \bm \Sigma_r ^ {1/2} $.
    When $r < K$, the dimension of $\bm F_u$ is less than that of $\bm F$.
    
    
   Although both linear conic optimization and bisection search are solvable and have polynomial time complexity, below we provide a more efficient approach for \eqref{eq-balancing-dpc}. 
    
    \subsection{Dual program solution} \label{sec-dpc-b}
    
    This section derives the dual problem of \eqref{eq-balancing-dpc}, and presents a gradient projection based method to solve it. 
    
    \subsubsection{Formulation of the dual problem}
   The optimal value of \eqref{eq-balancing-dpc} is equal to that of the following dual program:  
        \begin{subequations} \label{eq-sinr-balancing-dual}
        	\begin{align}    
            \min_{\bm Y \succeq 0}   \min_{ \bm d \geq 0, \gamma>0} \ \gamma ,  \ \  
            \mathrm{s.t.} \ \sum_{k=1}^{K} d_k &\geq 1,   \\
        		\mathrm{tr} (\bm Y  ) & = 1, \label{eq-dual-cond-a} \\
        		\bm Y+  \sum_{i<k} d_i \bm u_i \bm u_i ^ H &\succeq \frac{1}{\gamma} d_k \bm u_k \bm u_k ^ H, \ k =1,\ldots,K. \label{eq-dual-cond-b}
        	\end{align}
        \end{subequations}    
     The proof starts from the Lagrange dual problem of the power minimization in \eqref{eq-minimal-power2}, and derives the dual problems of \eqref{eq-minimal-power2} and \eqref{eq-balancing-dpc}, sequentially, as will be detailed in Appendix~\ref{A-6}.

    Note that \eqref{eq-sinr-balancing-dual} is still nonconvex, since the constraints in  \eqref{eq-dual-cond-b} are bilinear. 
    We propose an iterative method to solve it, summarized later in Sec.~\ref{subsub:fixed-point and gradient}, where the inner optimization problem is solved by fixed-point method and the outer is solved by gradient decent. 
    
    \subsubsection{Fixed-point iterations for inner minimization} 
    Under a given $\bm Y$, the inner minimization of \eqref{eq-sinr-balancing-dual}  is stated as
    \begin{equation} \label{eq-opty}
        \min_{ \bm d \geq 0, \gamma>0} \ \gamma ,  \ \  
        \mathrm{s.t.} \  \sum_{k=1}^{K} d_k \geq 1  \  \mathrm{and}  \ \eqref{eq-dual-cond-b}.
    \end{equation}
    Define the index set $\mathcal{I}(\bm Y) = \{ k \,  | \,  \bm u_k \notin  \mathcal{C}(\bm Y) , \ 1 \leq k \leq K \}$.
    The following equations should hold at the optimum of \eqref{eq-opty}:
    \begin{subequations} \label{eq-dd}
        \begin{align}
            &   \sum_{k = 1}^{K} d_k = 1, \   d_k = 0 \ \mathrm{for} \  k \in  \mathcal{I}(\bm Y), \\ & \gamma = d_k \bm u _ k ^ H (\bm Y + \sum_{i<k} d_i \bm u_i \bm u_i ^ H ) ^ \dagger \bm u_k \ \mathrm{for} \  k \notin  \mathcal{I}(\bm Y). \label{eq-dd-b}
        \end{align}
    \end{subequations}
    For the case of $k \notin \mathcal{I}(\bm Y)$, according to \cite{wiesel_linear_2006,4203115} we compute  $\{  d_k\}$ by the fixed-point iterations  given in Algorithm~\ref{fix-point}, which converges rapidly in practice. After $\{ d_k \}$ is obtained, $\gamma$ can be computed via \eqref{eq-dd-b}. 
    
    
    \begin{algorithm}
        \renewcommand{\algorithmicrequire}{\textbf{Input:}}
        \renewcommand{\algorithmicensure}{\textbf{Output:}}
        \caption{Fix-point algorithm to solve \eqref{eq-opty}}
        \label{fix-point}
        \begin{algorithmic}[1]
            \STATE Initialization:  $ \{ d_k  > 0 \}$.
            \REPEAT
            \STATE Compute $\hat{d}_k$ from \eqref{eq-dd-b}, for all $k \notin  \mathcal{I}(\bm Y)$:  
            $
             \hat{d}_k   \leftarrow \Big[  \bm u_k ^H ( \bm Y + \sum_{i<k} d_i  \bm u_i \bm u_i ^ H ) ^ \dagger  \bm u_k  \Big]^{-1}.
            $
            \STATE Store the value of $\{ d_k \}$, for all $k \notin  \mathcal{I}(\bm Y)$: $  d_k' \leftarrow d_k  $.
            \STATE Normalize $\{d_k\}$ so that their sum is $1$: 
          $ d_k  \leftarrow \frac{\hat{d}_k }{ \sum_{k = 1}^{K}\hat{d}_k },  \ \forall k \notin  \mathcal{I}(\bm Y). $
            \UNTIL $\sum_{k \notin \mathcal{I}(\bm Y)  }  | d_k - d_k '|^2 < \varepsilon $.  
            \ENSURE  Solution of  $d_k$, for $ k \notin  \mathcal{I}(\bm Y)$.
        \end{algorithmic}  
    \end{algorithm}
    
    \subsubsection{Gradient derivation for the outer minimization}
    Writing the optimal value of \eqref{eq-opty} as $\gamma_o(\bm Y)$,  we reformulate \eqref{eq-sinr-balancing-dual} into an optimization with respect to $\bm Y$:
    \begin{equation} \label{eq-sinr-balancing-dual-2}
        \min_{\bm Y \succeq 0}    \ \gamma_o(\bm Y)  ,  \ \  
        \mathrm{s.t.} \ \eqref{eq-dual-cond-a}.  
    \end{equation}
    We observe that, when $\bm Y \succ 0$, $\gamma_o(\bm Y)$ is a convex and differential function.
    Therefore, \eqref{eq-sinr-balancing-dual-2} can be solved via gradient projection. 
    
    We first give the expression of the gradient $ \nabla \gamma_o(\bm Y)  $ below, which we note that is non-trivial because there is no analytical expression of the objective function $  \gamma_o(\bm Y)  $. 
    For readability, we leave the derivation in Appendix~\ref{A-4}. 
    We define 
    \begin{equation} \label{eq-fk}
        \hat{\bm f}_k^*(\bm Y) = (\bm Y + \sum_{i<k} d_i ^ * (\bm Y) \bm u_i \bm u_i ^ H) ^ {-1} \bm u_k,
    \end{equation}
    where $ {d_k^*(\bm Y)} $ is the solution of \eqref{eq-opty} and \eqref{eq-dd}, and a $K \times K$ lower triangular matrix $\bm A$ by
    \begin{equation} \label{eq-A}
        \bm A_{k,i}  = -   d_k^* (\bm Y)|\bm u_i ^ H \hat{\bm f}_k^*(\bm Y)  |^2 , \
        \bm A_{k,k} = \bm u_k ^ H \hat{\bm f}_k^*(\bm Y)  , 
    \end{equation}
    for $1 \leq k \leq K$ and $1 \leq i < k$. 
    Letting $\bm a = ( \bm A^{-1} )^T \bm 1_K $, we have
    \begin{equation} \label{eq-gradY}
        \nabla \gamma_o(\bm Y)  = - \frac{\sum_{k=1}^{K}  a_k d_k^*  (\bm Y) \hat{\bm f}_k^*(\bm Y) 
            [ \hat{\bm f}_k^*(\bm Y)] ^ H}{\sum_{k=1}^{K}  a_k }, 
    \end{equation}
    where $a_k$ is the $k$-th element in $\bm a$.
    
    Under the constraint \eqref{eq-dual-cond-a}, the descend direction at a point $\bm Y  \succ 0$ is computed from a projection operation \cite{doi:10.1137/S1052623497330963}:
    \begin{equation}
        \bm  \Delta _y (\bm Y)= \mathcal{P}_{\Omega_2} \left( \bm Y - \nabla \gamma_o (\bm Y)    \right) - \bm Y.
    \end{equation}
    where  $ \mathcal{P}_{\Omega_2} (\cdot)$ is the orthogonal projection onto the constraint set $ \Omega_2 = \{  \bm Y \, | \,  \bm Y \succeq 0, \ \mathrm{tr} (\bm Y) = 1   \}$.
    The projection does not have a close form expression, but can be computed by at most $K$ loops; See Appendix~\ref{A-3}. 
    
    \subsubsection{Summary of the iterations for the dual problem}
    \label{subsub:fixed-point and gradient}
    
    With the initial value $\bm Y ^ {(0)}$, which should be strictly feasible, namely $\bm Y  ^ {(0)} \succ 0$ and $ \mathrm{tr} (\bm Y ^ {(0)}) = 1 $, we update $\bm Y$ and $\bm d$ by the following iterations for  $\ell =0,1,\ldots$
    \begin{enumerate}[(a)]
        \item Update $\bm Y$ via $\bm Y ^ {(\ell+1)} := \bm Y ^ {(\ell )} + \alpha ^ {(\ell )} \bm   \Delta _  y  (\bm Y ^ {(\ell )}). $
        \item 
        Update $\bm d$ and $\gamma$ via \eqref{eq-dd}, yielding $\bm d ^ {(\ell+1)} :=  \bm d^ * (\bm Y ^ {(\ell+1)} )$ and $\gamma ^ {(\ell+1)} :=  \gamma_o ( \bm Y ^{(\ell+1)}) $.
    \end{enumerate}
    Here, the step size $\alpha ^ {(\ell )} \in (0,1) $ can be determined by backtracking search \cite{boyd_vandenberghe_2004}.
    The iterations of gradient descend can be stopped if $ \| \bm  \Delta _  y  (\bm Y ^ {(\ell )})  \|_F $ is small enough.
    Let $L$ be the number of iterations at convergence.
    After \eqref{eq-sinr-balancing-dual} is solved, the balanced \gls{SINR} is  given by $ \gamma := \gamma ^ {(L)}$.

    \subsubsection{The solution for the primal problem}
    \label{subsub: sinr balancing: primal}
    
    With the solution of $\gamma_o$ obtained by the above dual problem, we then compute the solution of the primal problem in \eqref{eq-balancing-dpc}. In this regard, we first compute the optimum $\bm F_u$ for \eqref{eq-minimal-power2}. Then, the optimal $\bm F$ in \eqref{eq-balancing-dpc} is given by $\bm F := \bm U \bm \Sigma_r ^ {1/2} \bm F_u$.
    
    The relationship between the dual and primal problems is stated as follow. Given $\gamma =  \gamma_o^*$,  the optimum  of \eqref{eq-sinr-balancing-dual} should also be the solution of the dual problem of \eqref{eq-minimal-power2}. The dual of  \eqref{eq-minimal-power2} is provided in \eqref{eq-minimal-power2-dual} of the Appendix~\ref{A-6}. 
     From the \gls{KKT} conditions, the optimal solution of \eqref{eq-minimal-power2} and \eqref{eq-minimal-power2-dual}  obey
    \begin{equation} \label{eq-ff2}
    	\big(\bm Y+  \sum_{i<k} d_i \bm u_i \bm u_i ^ H \big ) \bm f_k = \frac{1}{\gamma_k} d_k \bm u_k \bm u_k ^ H \bm f_k, \ k =1,\ldots, K.
    \end{equation}   
    
    From \eqref{eq-ff2}, we compute  $\bm F_u$ by $\bm f_k := \sqrt{b_k} \tilde{\bm f}_k$, where 
    \begin{displaymath}
    	\tilde{\bm f}_k :=  \big(\bm Y ^{(L)} + \sum_{i < k} d_i ^{(L)} \bm u_i \bm u_i ^ H \big) ^ {\dagger} \bm u_k,
    \end{displaymath} 
    and $ b_k$ is a real factor, for $k = 1,\ldots,K$.
    The factors $ \{b_k \}$ can be determined by solving the following $K$ linear \gls{SINR} equations:
    \begin{equation}
    	\frac{1}{\gamma^{(L)}} \left|\bm u_k ^ H \tilde{\bm f}_ k\right|^2 b_k - \sum_{i > k}  \left|\bm u_k ^ H \tilde{\bm f}_i\right|^2 b_i = 1 , \ k =1,\ldots,K.
    \end{equation}

    
   

    \section{Sum rate maximization for DPC} \label{sec-5}
    
    In this section, we solve the sum rate maximization in the DPC regime, 
    given by
    \begin{subequations} \label{eq-max-dpc_sumrate}
        \begin{align}
            \max _{\bm F,\bm \gamma} \ & \  \sum_{k=1}^{K} \log(1+\gamma_k), \ 	\mathrm{s.t.}  \ \bm F \bm F^H \preceq \bm R_h \ \mathrm{and}    \\
            &   \   \frac{1}{\gamma_k}  |\bm F_{k,k} | ^ 2 \geq  \sum_{i > k}  | {\bm F}_{k,i} |^2 +  1 , \ k = 1,\ldots,K.  
        \end{align}
    \end{subequations}    
    This problem is from \eqref{eq-max-dpc1} with the target function being the sum rate. 
    Since the \gls{SINR} constraint is non-convex,  sum rate maximization is non-convex and  hard to solve.

    We solve \eqref{eq-max-dpc_sumrate} via the well-known downlink-uplink duality \cite{vishwanath_duality_2003, viswanath_sum_2003, 4203115}, which introduces a dual uplink \gls{MAC} that has the same achievable rate region as the downlink \gls{GBC}. 
    Then the problem becomes  the sum rate maximization for the \gls{MAC}. 
    From the duality, we are able to compute the optimal $\bm F$ via an equivalent convex optimization. 


    \subsection{Optimization reformulation based on downlink-uplink duality}
    
    This section first introduces the signal model of the dual uplink \gls{MAC} with respect to the original downlink model, given in Sec.~\ref{subsub: MAC}. 
    Based on such dual uplink \gls{MAC}, we formulate the sum rate maximization problem in  Sec.~\ref{susub: opt for MAC}.   
    Then in Sec.~\ref{subsub: duality}, we illustrate the  downlink-uplink duality by showing that the uplink \gls{MAC} and the downlink \gls{GBC} have the same achievable rate region. 
    Based on the  duality, Sec.~\ref{subsub: downlink} provides the solutions to the original downlink transmit design problem. 
    
    \subsubsection{Dual uplink \gls{MAC} model}
    \label{subsub: MAC}
    
    Consider a uplink \gls{MAC}, in which $K$  users simultaneously transmit to a base station with $r$  antennas.
    Each user is equipped with a single transmit antenna.
    The channel is $\bm \Sigma_r ^ {1/2} \bm U^H \in \mathbb{C} ^ {r \times K} $. The received signal  is
    \begin{equation}
        \bm y_\mathrm{ul} =  \bm  \Sigma_r ^ {1/2} \bm U^H  \bm x_\mathrm{ul} + \bm v_\mathrm{ul},
    \end{equation}
    where $  \bm x_\mathrm{ul} = [ x_{\mathrm{ul},1},\ldots, x_{\mathrm{ul},K} ]^T $ includes the transmit signal of the users,  and $\bm v_\mathrm{ul}$ is additive Gaussian noise that has uncertain covariance $ \bm Y$ constrained by $\mathrm{tr} (\bm Y  ) = 1$, analogy to \eqref{eq-dual-cond-a}. 
    The transmit power of the $k$-th user, denoted by $ d_k$, for $k  =1,\ldots,K$, should be optimized, under the sum power constraint $ \sum_{k=1} ^ K d_k \leq 1$. 
    
    For the $k$-th user, the receiver applies a linear filter $ \hat{\bm f}_k $, and the output is
    \begin{equation}
        \hat{\bm f}_k ^ H \bm y_\mathrm{ul} =   \sum_{i=1}^{K} \hat{\bm f}_k ^ H \bm u_i x_{\mathrm{ul},i} +  	\hat{\bm f}_k ^ H  \bm v_\mathrm{ul}.
    \end{equation}
    Corresponding to the \gls{DPC} strategy in the downlink regime, we use successive cancellation with the reverse  order $\{K,\ldots,1\}$  \cite{viswanath_sum_2003} so that the signal from the $k+1,\ldots,K$-th user can be subtracted when decoding for the $k$-th user. 
    Therefore, the \gls{SINR} for the $k$-th user is
    \begin{equation}
    \label{equ:sinr mac}
        \mathrm{SINR}_{k}^{\mathrm{mac}} =  \frac{d_k | \hat{\bm f}_k^H \bm u_k |^2}{  \sum_{i<k} d_i | \hat{\bm f}_k^H \bm u_i |^2 + \hat{\bm f}_k^H \bm {Y} \hat{\bm f}_k }.
    \end{equation}
    
    The sum rate is then given by $\sum_{k = 1}^K \log ( 1 +  \mathrm{SINR}_{k}^{\mathrm{mac}} )$. 
    Below, we formulate the sum rate maximization for the \gls{MAC} with respect to the filters $\{ \hat{\bm f}_k \}$, the transmit power $\bm d := [d_1,\dots,d_K]^T$, and the noise covariance $\bm Y$. 
    
    \subsubsection{Sum rate maximization formulation}
    \label{susub: opt for MAC}
    
    The maximum sum rate is defined with regard to the worst case of noise: seeking for the noise variance  $\bm Y$ that most worsens the sum rate. Therefore, we only need to consider non-singular $\bm Y$, i.e., $\bm Y \succ 0 $, because the  sum rate can be infinity when $\bm Y$ is singular. 
    
    When $\bm Y$ is non-singular, the  \gls{MMSE} filter  that maximizes the output \gls{SINR} for the $k$-th user is given by  \cite{4203115}
    \begin{equation}\label{equ:mmse} 
        \hat{\bm f}_k = (\bm Y + \sum_{i<k} d_i \bm u_i \bm u_i ^ H) ^ {-1} \bm u_k, \ k =1,\ldots, K.
    \end{equation}
    Correspondingly,  the achieved \gls{SINR}  for the $k$-th user in \eqref{equ:sinr mac} becomes
    \begin{equation}
    \label{equ:sinr mac mmse}
        \mathrm{SINR}_{k}^{\mathrm{mac}}  = d_k \bm u_k ^H (\bm Y + \sum_{i<k} d_i \bm u_i \bm u_i ^ H) ^ {-1} \bm u_k , \ k =1,\ldots, K,
    \end{equation}
    and further the sum rate is  written as
    \begin{equation} \label{eq-rk}
        \sum_{k = 1}^K \log ( 1 +  \mathrm{SINR}_{k}^{\mathrm{mac}} ) = \log \big| \bm Y+  \sum_{k=1}^{K} d_k \bm u_k \bm u_k ^ H \big | - \log   \big| \bm Y \big|.
    \end{equation}  
    
    Now we  write the  sum rate maximization for the dual \gls{MAC} as
    \begin{subequations} \label{eq-weighted-rate-dual}
        \begin{align}
            \min_{\bm Y \succeq 0} &  \max_ {\bm d \geq 0}   \   \log \big| \bm Y+  \sum_{k=1}^{K} d_k \bm u_k \bm u_k ^ H \big | - \log   \big| \bm Y \big|   \\  &  \mathrm{s.t.}  \  \mathrm{tr} (\bm Y ) = 1, 
            \ \sum_{k=1}^{K} d_k = 1.
        \end{align}
    \end{subequations}
    In \eqref{eq-weighted-rate-dual}, all the constraints are convex.
    The objective function, denoted by $g(\bm Y, \bm d)$, is convex in  $ \bm d$ and is concave in $\bm Y$.
    Therefore, \eqref{eq-weighted-rate-dual} is a convex-concave saddle point problem. 
    Solutions to this problem will be discussed later in Sec.~\ref{subsec: equivalent minimization}.

    \subsubsection{Downlink-uplink duality}
    \label{subsub: duality}
    
    
    
    The downlink-uplink duality is established from the power minimization problem and its dual.
    In \eqref{eq-minimal-power2}, we assign  individual \gls{SINR} thresholds to the users, and the power minimization becomes:
    \begin{equation} \label{eq-minimal-power2-individual}
        \min_{\lambda, \bm F_u} \ \lambda  , \  \ \mathrm{s.t.} \ \bm F_u \bm F_u^H \preceq \lambda \bm I_r,  \ \frac{1}{\gamma_k}  |\bm u_k ^ H \bm f_k | ^ 2 \geq  \sum_{i > k}  | \bm u_k ^ H \bm f_i |^2 +  1 , \ k = 1,\ldots,K,
    \end{equation}
    where $\gamma_k$ is the given \gls{SINR} for the $k$-th user, for $k = 1,\ldots,K$.
    Correspondingly, the Lagrange dual problem becomes:
   \begin{equation} \label{eq-minimal-power2-dual-individual}
       \max_{\bm Y \succeq 0} \ \max_ {\bm d \geq 0}  \ \sum_{k=1}^{K} d_k, \ \ \mathrm{s.t.} \ \eqref{eq-dual-cond-a}, \ \mathrm{and} \ \bm Y+  \sum_{i<k} d_i \bm u_i \bm u_i ^ H \succeq \frac{1}{\gamma_k} d_k \bm u_k \bm u_k ^ H, \ k =1,\ldots,K.	
   \end{equation}
   We  denote the optimal value of \eqref{eq-minimal-power2-individual} and \eqref{eq-minimal-power2-dual-individual} by  $ \lambda^*( \bm \gamma) $ and $\lambda^*_\mathrm{dual}( \bm \gamma)$, respectively, where $\bm \gamma = [ \gamma_1, \ldots, \gamma_K] ^T$.


    Recall that for the \gls{GBC}, the \glspl{SINR} $\gamma_1, \ldots, \gamma_K$ are achievable  if and only if $ \lambda^*( \bm \gamma) \leq 1  $. 
    Meanwhile, we note that the inner maximization in \eqref{eq-minimal-power2-dual-individual} is equivalent to \cite{4203115}
    \begin{equation}
        \min _{\bm d \geq 0, \{ \hat{\bm f}_k \}}  \   \sum_{k=1}^{K} d_k , \ 
        \mathrm{s.t.}  \ \mathrm{SINR}_{k}^{\mathrm{mac}} \geq \gamma_k, \ k =1, \ldots, K,
    \end{equation}
    which finds the minimal  transmit power of the \gls{MAC} to achieve the \glspl{SINR} under a given $\bm Y$.
    Further, the optimal value of the outer maximization in \eqref{eq-minimal-power2-dual-individual} is the worst-case minimal transmit power under all possible $\bm Y$ constrained by \eqref{eq-dual-cond-a}.
    Therefore, $\lambda^*_\mathrm{dual}( \bm \gamma)$ gives the minimal transmit power to achieve the  \glspl{SINR} $\gamma_1, \ldots, \gamma_K$ in the \gls{MAC}. 
    Since the transmit power cannot exceed $1$ in the \gls{MAC}, the \glspl{SINR} $\gamma_1, \ldots, \gamma_K$ are achievable in the \gls{MAC} if and only if $ \lambda^*_\mathrm{dual}( \bm \gamma) \leq 1  $.
    From strong duality, $ \lambda^*_\mathrm{dual}( \bm \gamma)  = \lambda^*( \bm \gamma) $, so the achievable region of the \gls{GBC} and \gls{MAC} are the same.

    \subsubsection{Solutions to the original downlink problem}
    \label{subsub: downlink}    

    We compute $\bm F$ for the original downlink problem after the saddle point $(\bm Y^*, \bm d ^*)$ is obtained.
    First, from the downlink-uplink duality,  the obtained  \glspl{SINR}  in the \gls{MAC}, given by
     \begin{equation} \label{eq-sinr-ul}
    	\gamma_{\mathrm{ul}, k} :=  d_k ^ * \bm u ^ H \big(\bm Y ^ * + \sum_{i<k} d_i ^* \bm u_i \bm u_i ^ H \big) ^ {-1} \bm u_k, \ k =1, \ldots, K,
    \end{equation}
     also give the \glspl{SINR} in  the \gls{GBC}.     
    Then, with the known \glspl{SINR}, $\bm F$ is obtained by solving $\bm F_u$ from  \eqref{eq-minimal-power2-individual}, i.e., $\bm F := \bm U \bm \Sigma_r ^ {1/2} \bm F_u$, analogy to Sec.~\ref{subsub: sinr balancing: primal}. 
    In particular, we compute  $\bm F_u$ by $\bm f_k := \sqrt{b_k} \tilde{\bm f}_k$, where 
    \begin{displaymath}
        \tilde{\bm f}_k :=  \big(\bm Y ^ * + \sum_{i < k} d_i ^ *  \bm u_i \bm u_i ^ H \big) ^ {-1} \bm u_k,
    \end{displaymath} 
    and $ b_k$ is a real factor, for $k = 1,\ldots,K$.
    The factors $ \{b_k \}$ can be determined by solving the following $K$ linear \gls{SINR} equations:
    \begin{equation}
        \frac{1}{\gamma_{\mathrm{ul}, k}} \left|\bm u_k ^ H \tilde{\bm f}_ k\right|^2 b_k - \sum_{i > k}  \left|\bm u_k ^ H \tilde{\bm f}_i\right|^2 b_i = 1 , \ k =1,\ldots,K.
    \end{equation}

    \subsection{Solutions to \eqref{eq-weighted-rate-dual}}
    \label{subsec: equivalent minimization}
    
    To our knowledge, there are three types of methods to solve the convex concave saddle point problem in \eqref{eq-weighted-rate-dual}:
    \begin{enumerate}[(a)]
    	\item The first type is first order algorithms, such as  extra-gradient  and optimistic
    	gradient descent ascent \cite{pmlr-v108-mokhtari20a}, which only require the gradient;
    	\item The second type is interior-point algorithms \cite{boyd_vandenberghe_2004, 4203115}, which solve the \gls{KKT} equations via Newton method and thus require the second derivative;
    	\item The third one, as stated in \cite{doi:10.1080/10556788.2021.1928121}, is to convert the saddle point problem to  equivalent linear conic problems that is acceptable to convex optimization solvers like CVX \cite{gb08,cvx}.
    \end{enumerate}
    
    The implementation of the interior-point algorithms and first order algorithms is omitted here.
    Considering the well-structure of the saddle point problem in \eqref{eq-weighted-rate-dual}, we show that it is equivalent to the following convex optimization:
    \begin{equation} \label{eq-optz}
        \min _{\bm Z \succeq 0}  \  \log | \bm I_r +  {\bm Z} | - \log |\bm Z|  , \
        \mathrm{s.t.} \ \bm u _k ^ H \bm Z \bm u_k \leq 1, \ \forall k.
    \end{equation}
    Here, the optimal $\bm Z$ should be non-singular. 
    The equivalence between  \eqref{eq-weighted-rate-dual} and \eqref{eq-optz} is from the following theorem.

    \begin{mytheorem} \label{theorem-1}
        Let ${\bm Z}^*  $ be the optimum  of  \eqref{eq-optz}. Hence, ${\bm Z}^*  $ should obey
        \begin{equation} \label{eq-Zphi}
            {\bm Z} ^ {*-1} - ( \bm I_r + \bm Z^* )^ {-1} =\sum_{k=1}^{K} \phi_k \bm u_k \bm u_k ^ H, \ \phi_k  \bm u_k ^ H {\bm Z} ^* \bm u_k = \phi_k,
        \end{equation}  
        for $k  =1,\ldots,K$, 	where $\{ \phi_k \geq 0 \}$ are the dual variables.
        Then the saddle point of  \eqref{eq-weighted-rate-dual} can be computed by
        \begin{equation} \label{eq-dY}
            \ \bm d^* = \frac{1}{\eta} \bm \phi, \  {\bm Y}_1 ^ * =  \frac{ 1}{\eta} (  {\bm Z} ^ * +  \bm I_r )^{-1},
        \end{equation}
        where $ \bm \phi = [ \phi_1, \ldots, \phi_K] ^ T $ and $\eta =  \bm 1_K ^ T \bm \phi$.
    \end{mytheorem}
    
    \begin{proof}
        The \gls{KKT} condition of \eqref{eq-optz} directly yields equations in \eqref{eq-Zphi}. 
        To show that $(\bm Y^*, \bm d ^*)$ is a saddle point of \eqref{eq-weighted-rate-dual}, we only need to verify  the  \gls{KKT} conditions; See Appendix~\ref{A-1}.
    \end{proof}
    
    The problem in  \eqref{eq-optz} can be reformulated to a linear conic optimization; See  \cite{kim_optimization_2019}.
    After $\bm Z^*$ is obtained, we can first compute $\{\phi_k\}$ from \eqref{eq-Zphi}, and then compute $(\bm Y^*, \bm d ^ *)$ from \eqref{eq-dY}.
    
    
    
    \subsection{Discussions}
    
    It is worth noting that the maximized sum rate in the  \gls{DPC} regime should equal to the sum rate capacity.
        To see this, we introduce a new variable $ \bm Z' = \bm U \bm \Sigma_r ^ {1/2} \bm Z \bm \Sigma_r ^ {1/2}  \bm U  ^ H$, and then \eqref{eq-optz} becomes
          \begin{equation} \label{eq-sato}
        	\min _{\bm Z' \succeq 0}  \  \log | \bm R_h + \bm Z' | -  \log |  \bm Z' | , \
        	\mathrm{s.t.} \  \bm Z_{k,k}' = 1, \  k = 1,\ldots,K,
        \end{equation}
        when $\bm R_h$ is non-singular.
        In \eqref{eq-sato}, the inequality constraints becomes equality constraints since the equality should hold at the optimum.
        The optimal value of \eqref{eq-sato}, named Sato upper bound \cite{vishwanath_duality_2003,viswanath_sum_2003,1327794}, gives the upper bound for the sum rate capacity of the \gls{GBC}.
        Note that the maximized sum rate via \eqref{eq-max-dpc_sumrate} is equal to the optimal value of \eqref{eq-optz}, and thus equals to the Sato upper bound.
        Therefore, the sum rate capacity is achieved by \gls{DPC}, which corresponds with the conclusion in \cite{weingarten_capacity_2006} that \gls{DPC} achieves the capacity region of  \gls{GBC} with the transmit covariance constraint.
        
%
    

    \section{Numerical results} \label{sec-6}
    
    We performed numerical simulations to demonstrate the performance of multiuser communications under the transmit covariance constraint from radar.
    The simulation settings are introduced in Sec.~\ref{sec-simu1}.
    In Sec.~\ref{sec-simu2}, the simulation results for \gls{SINR} balancing in the transmit beamforming and \gls{DPC} regimes are compared.
    The results of  sum rate maximization is displayed in Sec.~\ref{sec-simu3}.
    The convergence property of the iteration method  proposed in Sec.~\ref{sec-3} and~\ref{sec-4} are displayed in Sec.~\ref{sec-simu4}.

    \subsection{Preliminaries} \label{sec-simu1}
    
    In the simulations, the transmit array is a uniform linear array with equal antenna spacing.
    The antenna spacing is half of the wavelength, and the number of transmit antenna is $M = 10$.
    The optimal covariance for radar $\bm R_o$ is given by $ \bm R_o = P \bm S_o$, where $P$ is the transmit power and $\bm S_o$ is the power normalized covariance.
    For a given $\bm S_o$, we performed numerical experiments with different $P$ to obtain the communication performance versus transmit \gls{SNR} $P / \sigma^2$.
    We also compared the communication performance under three different values of $\bm S_o$, which corresponds to three different radar transmit beam patterns.
    The first value is $\bm S_o = \bm I_M / M$, with which the array transmits orthogonal waveforms and forms an omni-directional beam pattern for radar.
    The second value is $\bm S_o = (1/M) \bm 1_M \bm 1_M ^ T$, which means that the array works in phase-array mode and forms a single  beam towards $ 0 ^ \circ$.
    The third value is obtained via the beam pattern matching design in \cite{stoica_probing_2007} to form multiple beams  towards $ -40 ^ \circ, 0 ^ \circ, 40 ^ \circ$ with a beam width of $10 ^ \circ$.
    The corresponding transmit beam patterns under the three values of $\bm S_o$ are displayed in Fig.~\ref{fig:beams}.
    
    For communications, the channel $\bm H$ obeys Rayleigh fading, namely the elements in $\bm H$ satisfy independent standard complex normal distributions.
    The noise power is $\sigma ^ 2 = 1$.
    To display the communication performance, we run Monte Carlo tests with randomly generated $\bm H$, and computed the average performance of communications.

    \begin{figure}
        \centering
        \includegraphics[width=0.48\linewidth]{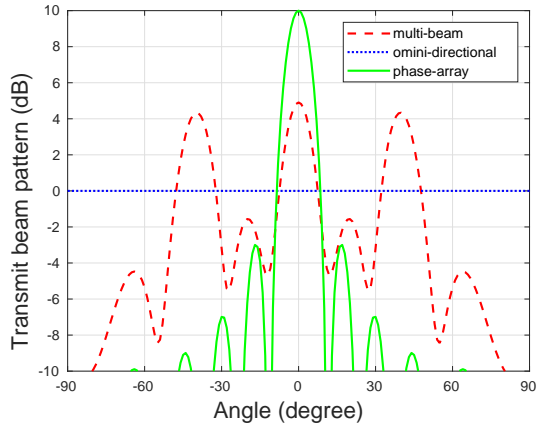}
        \caption{Three different types of transmit beam patterns for radar: omni-directional, multi-beam and phased-array single beam.}
        \label{fig:beams}
    \end{figure}
    
    \subsection{Balanced SINR versus transmit SNR} \label{sec-simu2}
    
     \begin{figure}
        \centering
        \begin{minipage}[t]{0.48\linewidth}
         \centering
        \includegraphics[width=\linewidth]{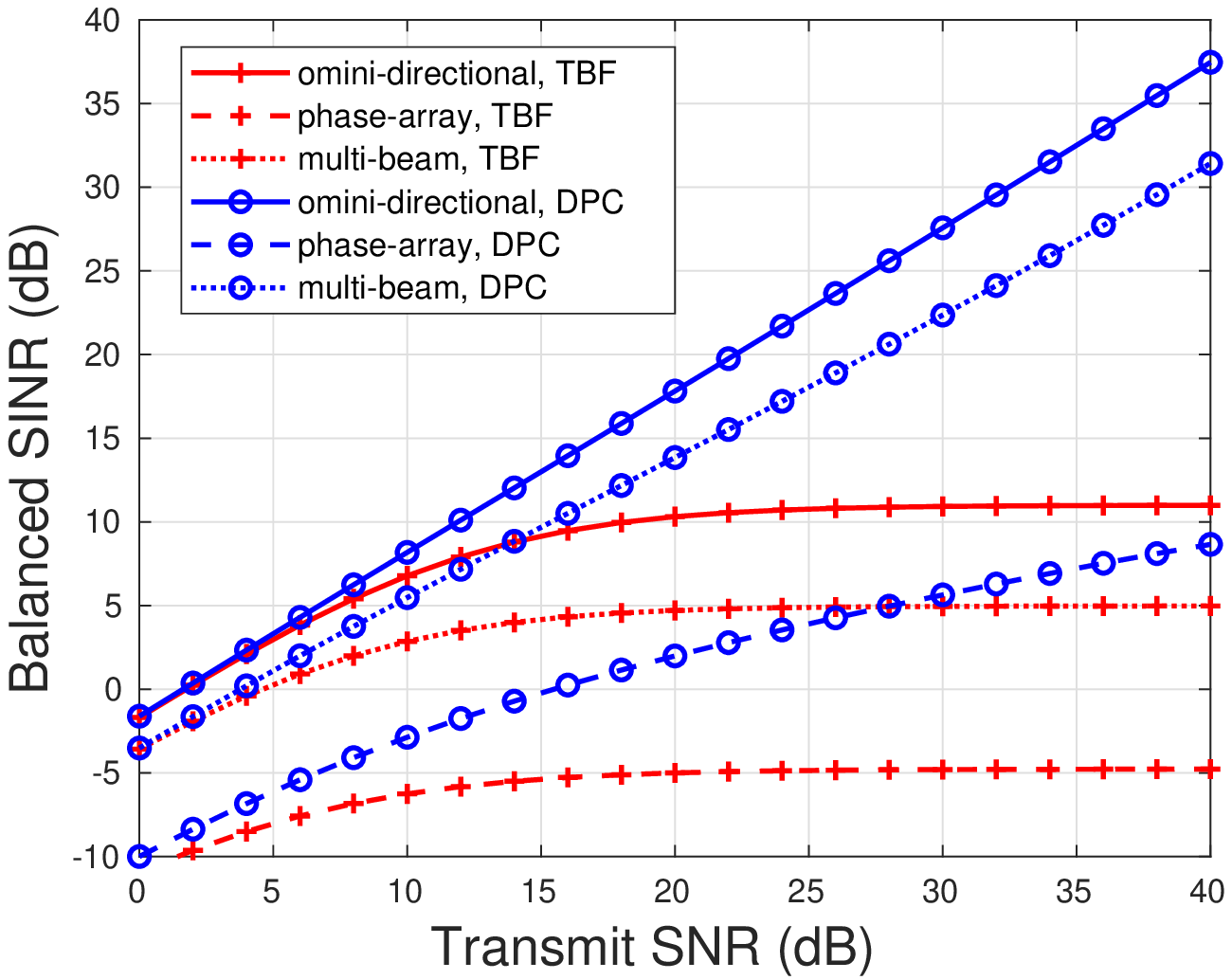}
        \caption{Balanced SINR versus transmit SNR $P / \sigma ^ 2$, for $K=4$. "TBF" is the curve for transmit beamforming.}
        \label{fig:p1}
         \end{minipage}
                  \hfill 
        \begin{minipage}[t]{0.48\linewidth}
         \centering
        \includegraphics[width=\linewidth]{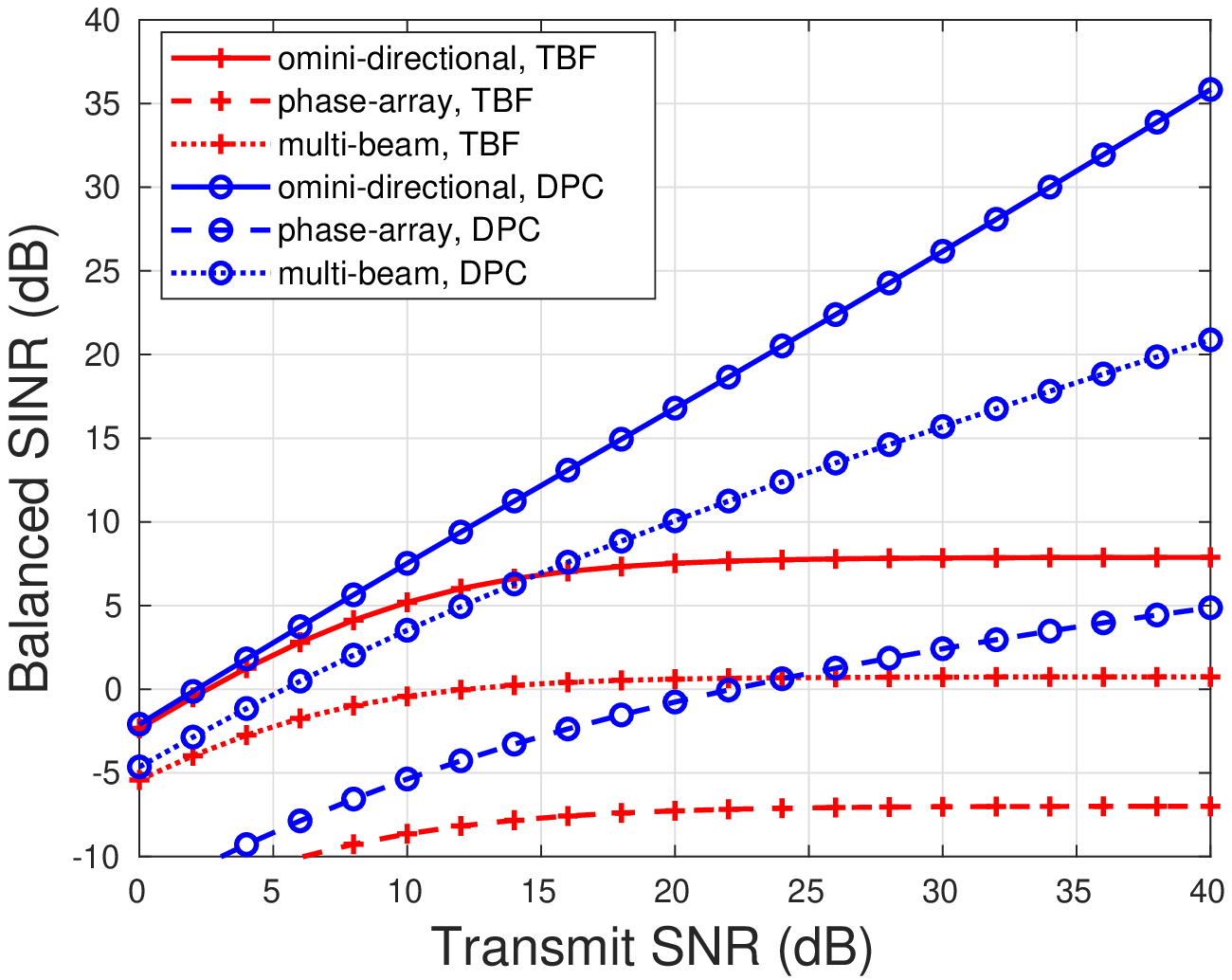}
        \caption{Balanced SINR versus transmit SNR $P / \sigma^2$, for $K=6$. "TBF" is the curve for transmit beamforming.}
        \label{fig:p2}
         \end{minipage}
    \end{figure}   
    
    The balanced \gls{SINR}  under different transmit \gls{SNR} and $\bm S_o$ for $K=4$ is displayed in Fig.~\ref{fig:p1}, in both transmit beamforming and \gls{DPC} regimes.
    From Fig.~\ref{fig:p1}, it is observed that \gls{DPC} achieves higher balanced \gls{SINR} than transmit beamforming.
    The performance improvement of \gls{DPC} over transmit beamforming is especially impressive under a high transmit \gls{SNR}.
    For omni-directional and multi-beam patterns, the balanced \gls{SINR} via  \gls{DPC}  increases linearly with the transmit \gls{SNR}  in dB scale, while the counterpart for transmit beamforming does not increase when the transmit \gls{SNR} is high.
    The reason is that the  interference cannot be effectively canceled via transmit beamforming with the transmit covariance constraint.
    To zero-forcing the interference, transmit beamforming  requires $\bm S_h = \bm H \bm S_o \bm H^H$ to be a diagonal matrix \cite{9124713}, while this condition generally does not hold if  $\bm H$ is Rayleigh fading.
    Since the interference cannot be eliminated, the balanced \gls{SINR} for transmit beamforming keeps constant even if the  \gls{SNR} is high.
    Conversely, \gls{DPC} are still able to cancel the interference under the transmit covariance constraint.
    As explained in Sec. \ref{dpc}, to zero-forcing the interference in the \gls{DPC} regime, we only need $\bm S_h$ to be non-singular.
    This condition can be  met if the rank of $\bm S_o$, denoted by $r_o$, is not less than $K$ when $\bm H$ obeys  Rayleigh fading.
    We note that the value of $r_o$ are $10$, $4$ and $1$ for the omni-directional,  multi-beam and phased-array pattern, respectively.
    When $K = 4$, the condition holds for  omni-directional and multi-beam patterns, and thus the corresponding balanced \gls{SINR} for  \gls{DPC}  is respectable under a high \gls{SNR}.
    For phased-array mode, $ r_o $ is less than $K$, and thus its  balanced \gls{SINR} in the \gls{DPC} regime becomes much lower compared to omni-directional and multi-beam patterns.
    Nevertheless, \gls{DPC} is still able to achieve an acceptable balanced \gls{SINR} for phased-array beam when the transmit \gls{SNR} is high, while we observe that the counterpart via transmit beamforming is even less than $0$-dB.
    
    The balanced \gls{SINR}  versus  transmit \gls{SNR} for $K = 6$ is shown in Fig.~\ref{fig:p2}.
    By comparing Fig.~\ref{fig:p1} and Fig.~\ref{fig:p2}, one  observes that the balanced \gls{SINR} becomes lower when $K$ increases, namely the service quality for  each user can  worsen when the number of users increases.
    We also observe that the loss of balanced \gls{SINR} in the \gls{DPC} regime is slight for omni-directional beam pattern, but is notable for the multi-beam pattern.
    We note that when $K  = 6$,   $ r_o$ is  larger than $K $ for omni-directional beam pattern, and thus \gls{DPC} is still able to zero-forcing the interference. 
    However, for the multi-beam pattern, $r_o$  is  less than $K = 6$, and thus \gls{ZF} \gls{DPC} is not applicable, leading to an obvious performance degradation.
    Based on the above facts, we can regard the $r_o$, the rank of $\bm S_o$, as the degrees of freedom of the communication transmitter.
    In the \gls{DPC} regime, increasing $K$ generally does not cause serious loss of service  quality  if $K$ does not exceed the degrees of freedom, while the loss can be more significant when $K$ exceeds it. 
    It is worth noting that when $K$ exceeds $r_o$,   simultaneously servicing for  $K$ users via transmit beamforming  is almost unrealistic, since the balanced \gls{SINR} is extremely low.

    \subsection{Maximized sum rate versus transmit SNR} \label{sec-simu3}

    \begin{figure}
        \centering
        \includegraphics[width=0.48\linewidth]{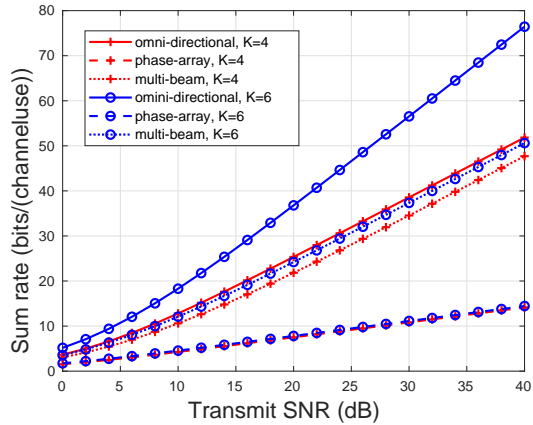}
        \caption{Maximized sum rate versus transmit SNR $P / \sigma^2$, for $K=4,6$. }
        \label{fig:p3}
    \end{figure}

    The maximized sum rate versus transmit \gls{SINR} in the \gls{DPC} regime is given in Fig.~\ref{fig:p3}, for $K = 4,6$.
    In  Fig.~\ref{fig:p3}, the sum rate is asymptotically affine in the transmit \gls{SNR} in dB, and the slope of the line determines the multiplexing  gain \cite{978730,lee_dirty_2006} of multiuser  communications, which equals to the rate gain in bits/channeluse for  every $3$-dB transmit power  gain. 
    With a power constraint, it is proven in \cite{lee_dirty_2006} that the multiplexing  gain of the \gls{GBC} is $K$. 
    With the considered transmit covariance constraint, the results is different.
    For instance, we read from Fig.~\ref{fig:p3} that the multiplexing  gain for the multi-beam pattern does not increase when $K$ increases from $4$ to $6$.
    A similar result is observed from the curve for phase-array mode.
    Nevertheless, for omni-directional beam pattern, the multiplexing  gain increases from $4$ to $6$ when $K$ increases from $4$ to $6$.
    In summary, one can find  that the multiplexing  gain is $\min\{K, r_o\}$ with the transmit covariance constraint.
    The explanation is from the fact that \gls{ZF} \gls{DPC} is asymptotic optimal under a high \gls{SNR}.
    When $K \leq r_o$, the \gls{GBC} can be simplified to $K$ \gls{AWGN} channels to the $K$ users via \gls{ZF} \gls{DPC}, and thus the multiplexing  gain is $K$.  
    However, when $K >  r_o$, to meet the constraint $ \bm F \bm F^H \preceq \bm R_h $,  $\bm F$ should have at most $r_o$ non-zero diagonal elements if it is lower triangular.
    In other words,  to  zero forcing the interference, only $r_o$ users are active while the others are inactive. 
    Therefore, when $K >  r_o$, the multiplexing  gain is restricted by  $r_o$, and the sum rate gain is not obvious if $K$ exceeds $   r_o$.

    \subsection{Convergence performance of the iteration algorithms} \label{sec-simu4}
    
    The convergence performance of the iteration algorithms to solve the \gls{SINR} balancing problem for transmit beamforming in Sec.~\ref{sec-bf-2} is displayed in Fig.~\ref{fig:c1}.
    In Fig.~\ref{fig:c1}, the \gls{SINR} gap $ \gamma ^ {(\ell)} - \gamma ^ * $ versus iteration time $\ell$ in $10$ experiments with randomly distributed $\bm H$ is demonstrated, for $K = 4$, $\bm S_o = (1/M) \bm I_M$ and $P  = 10$.
    Here, $ \gamma ^ *$ is the balanced \gls{SINR} obtained by the linear conic programming in \eqref{eq-sinr-p2}, and   $\gamma ^ {(\ell)}$ is the temporary \gls{SINR}  after the  $\ell$-th iteration, given by
    \begin{displaymath}
        t ^ {(\ell)} := \mathrm{{tr}} \Big\{ (\bm D  ^ {(\ell)} \bm [ D ^ {(\ell)}  ] ^ H    )^{1/2}   \Big\} , \  \gamma ^ {(\ell)} := [ t ^ {(\ell)}]^2 / ( 1 - [ t ^ {(\ell)}]^2) .
    \end{displaymath}
    
    From Fig.~\ref{fig:c1}, we observe that the iteration algorithm in Sec.~\ref{sec-bf-2} converges fast.
    In some experiments, the \gls{SINR} gap is less than $10^{-4}$ after no more than $10$ iterations.
    We note that the optimal $\bm d$ in \eqref{eq-dual-problem-d} may locate at the boundary, i.e. have zero elements.
    In this case, the algorithm may need more iterations to converge, as indicated by the curves in the right of Fig.~\ref{fig:c1}.
    Nevertheless, the algorithm can still find an acceptable approximate solution with a few iterations.
    Considering the interference control,  the number of users for transmit beamforming should be limited. 
    Therefore, the dimension of the variable $\bm d$ can be low in practice, and it is hopeful to implement the algorithm in real time.
    
    The convergence performance of the iterative algorithms to solve the \gls{SINR} balancing problem for \gls{DPC} in Sec.~\ref{sec-dpc-b} is displayed in Fig.~\ref{fig:c2}, which gives the \gls{SINR} gap $ \gamma ^ {(\ell)} - \gamma _o  ^ * $ versus iteration time $\ell$ in $10$ experiments.
    In each experiment, $\bm H$ is randomly generated with $K  =4$, and $\bm S_o = (1/M) \bm I_M$.
    To perform the experiments, we first let the balanced \gls{SINR} be $\gamma _o  ^ *  = 10$, next compute the minimal power $P'$ to achieve the \gls{SINR} by solving \eqref{eq-minimal-power2}, and then perform the iteration algorithm to solve \eqref{eq-sinr-balancing-dual} with the power $P'$, i.e. with $\bm R_o = P' \bm S_o$. 
    
    Compared with the the iteration algorithm for transmit beamforming, the algorithm for \gls{DPC} needs more iterations to achieve a small \gls{SINR} gap.
    This is mainly because the optimization in \eqref{eq-sinr-balancing-dual} has a  more complex structure and  a higher dimension than that in   \eqref{eq-dual-problem-d}.
    In the experiments, we observe that the iterations converge within a moderate number of times when the optimal $\bm Y$ is non-singular, while the \gls{SINR} gap decreases slower when the optimal $\bm Y$ is singular.
    For practical applications, a trade-off between the accuracy and computation time can be considered. 
    In other words, we can control the number of iterations and obtain an approximate solution.
    To improve the algorithm efficiency, one may further consider  deep leaning enabled acceleration schemes as in \cite{9027103}.

    \begin{figure}
        \centering
        \begin{minipage}[t]{0.48\linewidth}
         \centering
        \includegraphics[scale=0.48]{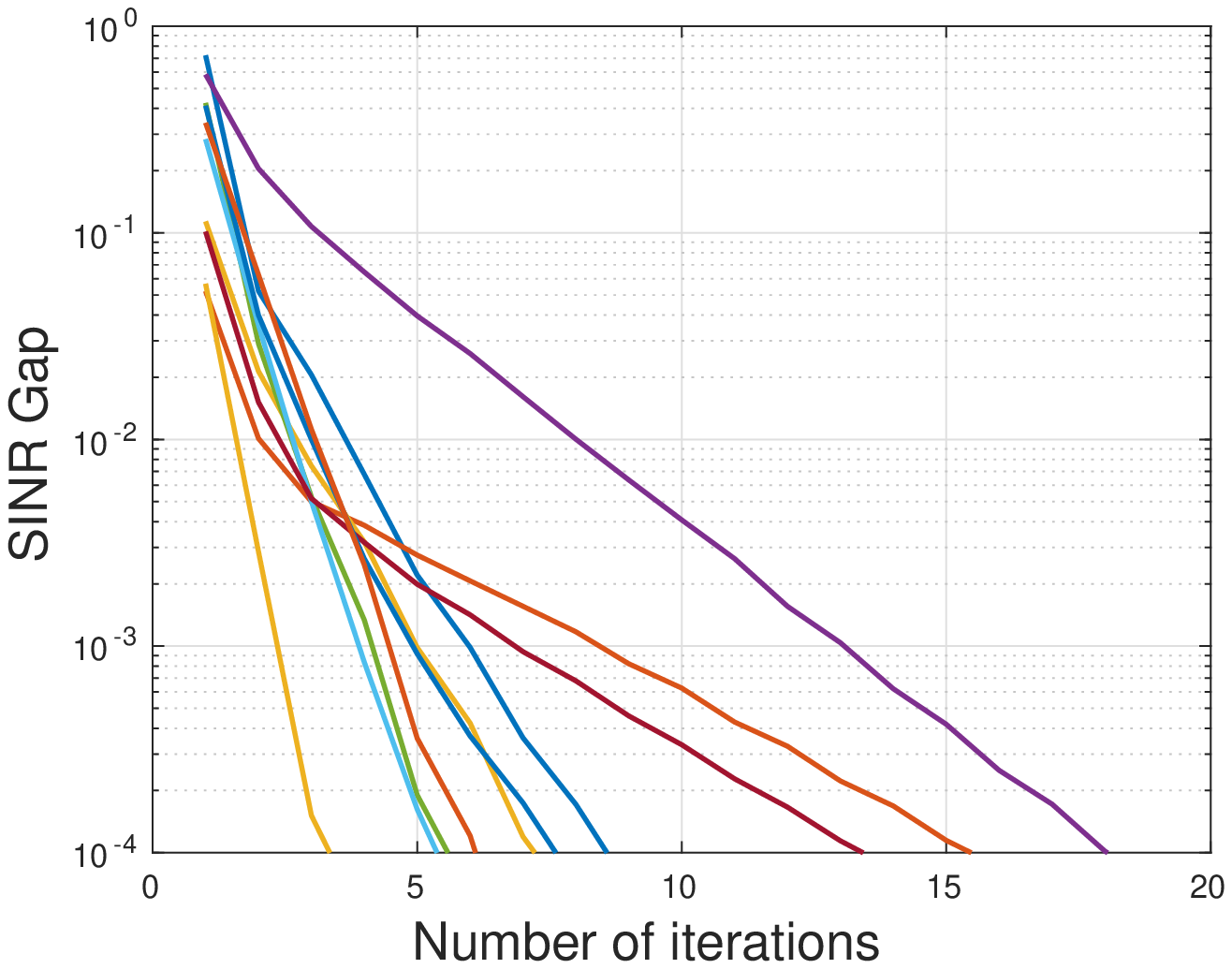}
        \caption{The \gls{SINR} gap $ \gamma ^ {(\ell)} - \gamma ^ * $ versus iteration time $\ell$ for  the \gls{SINR} balancing in the transmit beamforming regime.}
        \label{fig:c1}
        \end{minipage}
        \hfill
         \begin{minipage}[t]{0.48\linewidth}
        \centering
        \includegraphics[scale=0.48]{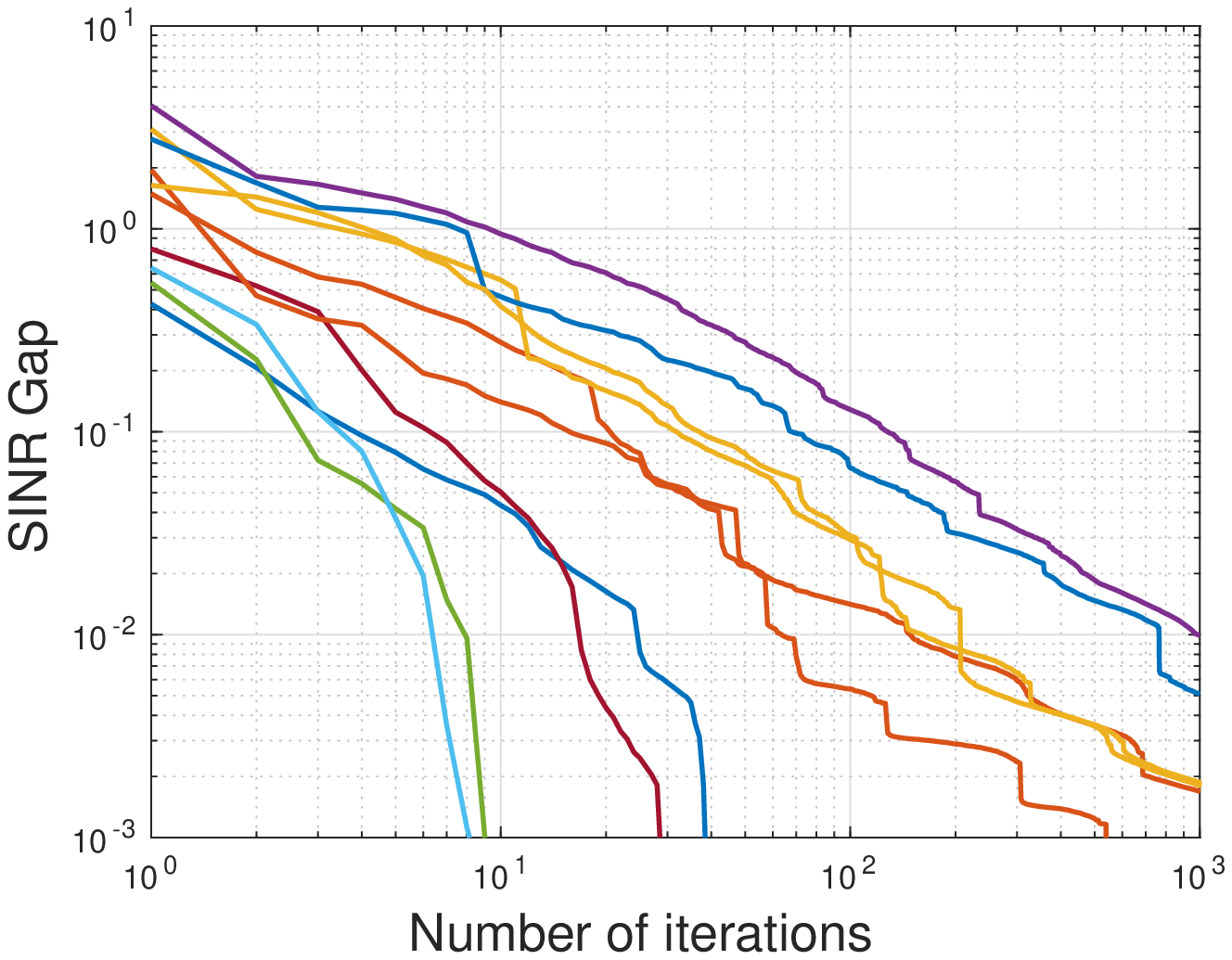}
        \caption{The \gls{SINR} gap $ \gamma ^ {(\ell)} - \gamma_o ^ * $ versus iteration time $\ell$ for  the \gls{SINR} balancing in the \gls{DPC} regime.}
        \label{fig:c2}
        \end{minipage}
    \end{figure}

    \section{Conclusion}  \label{sec-7}
    
    In this paper, we consider the transmit design of a joint \gls{MIMO} radar and downlink multiuser communications system, in which the communication performance is optimized under a transmit covariance constraint from radar.
    In particular, we formulate the \gls{SINR} balancing problem in both transmit beamforming and \gls{DPC} regimes, and the sum rate maximization in the \gls{DPC} regime.
    Further, we proposed methods to solve these problems via convex optimization.
    Despite the low complexity of transmit beamforming, the achievable \gls{SINR} via transmit beamforming may be low even if the transmit \gls{SNR} is high.
    As the theoretically optimal scheme for multiuser precoding, \gls{DPC} has a impressive performance gain over transmit beamforming, with increased complexity for encoding and optimization.
    In the simulations, it is observed that the degrees of freedom for the communication transmitter is restricted by the rank of the transmit covariance.
    
    \begin{appendices}

        \section{The Lagrange dual  of \eqref{eq-sinr-p3}} \label{A-5}
        The dual  objective function is defined by  $	g_1(\bm Y, \bm d) = \max _ {\bm F_u,t} \mathcal{L}_1 (\bm F_u,t, \bm Y, \bm d )$. 
        Under the conditions that
        \begin{equation} \label{eq-dual-condition1}
            \bm s^ T \bm d = 1, \  d_k \bm u_k \in \mathcal{C} (\bm Y) \ \mathrm{for} \ k =1,\ldots,K, 	
        \end{equation}
        where $\bm s = [s_1, \ldots, s_K]^T$, $	g_1(\bm Y, \bm d)$ is finite and its expression is
        \begin{equation}
            g_1(\bm Y, \bm d) = \mathrm{tr} \Big\{  \bm Y + \frac{1}{4} \bm Y ^ \dagger \bm D \bm D ^ H \Big \},
        \end{equation}
        where $\bm D = \left[ d_1 \bm u_1, \ldots, d_K \bm u_K \right]$.
        Thus, the  Lagrange dual problem of \eqref{eq-sinr-p3} is
        \begin{equation} \label{eq-dual-problem}
            \min _ {\bm Y, \bm d } \mathrm{tr} \Big\{ \bm Y + \frac{1}{4} \bm Y ^ \dagger \bm D \bm D ^ H \Big \} , \ \mathrm{s.t.}  \ \eqref{eq-dual-condition1}, \ \bm d \geq 0, \ \bm Y \succeq 0. 
        \end{equation}
        
        To solve the dual problem, we  consider to optimize $\bm Y$ with a given $\bm d$.
        From equation (2.4) in \cite{shebrawi_albadawi_2013}, one has 
        \begin{equation} \label{eq-trYRd}
            \mathrm{tr} 	\Big\{  \bm Y + \frac{1}{4} \bm Y ^ \dagger \bm D \bm D ^ H  \Big \}    \geq   \mathrm{tr} 	\Big\{	(   \bm Y  \bm Y ^ \dagger \bm D \bm D ^ H  )^{1/2} 	\Big\} =  \mathrm{tr}  	\Big\{	(   \bm D \bm D ^ H  )^{1/2}  	\Big\} ,
        \end{equation}
        where the inequality holds with equality when $\bm Y = \frac{1}{2} (   \bm D \bm D ^ H  )^{1/2} $.
        In \eqref{eq-trYRd}, we use the equality that $ \bm Y\bm Y ^ \dagger \bm D =  \bm D  $ since the columns of $ \bm D $ should be in $\mathcal{C}(\bm Y)$ according to \eqref{eq-dual-condition1}.
        Note that \eqref{eq-trYRd}  gives the optimal value  of \eqref{eq-dual-problem} under a given $\bm d$.
        Therefore, the dual problem in \eqref{eq-dual-problem}  is equivalent to \eqref{eq-dual-problem-d}.

        \section{Computation of the projection onto $\Omega_1$ } \label{A-2} 		
        The projection of a point $\bm d  = [d_1,\ldots,d_K] ^ T\in \mathbb{R} ^ K $ onto $\Omega_1$ is expressed as
        \begin{equation} \label{eq-optx}
            \bm x = \arg \min_{\hat{\bm x}} \frac{1}{2} \|  \bm d - \hat{\bm x}  \|_2 ^ 2, \ \mathrm{s.t.}  \ \hat{\bm x} \geq 0 , \ \bm s ^ T \hat{\bm x} = 1,
        \end{equation}
        where $\bm s  = [s_1, \ldots, s_K] ^ T> 0$.
        Letting $ \bm \theta = [ \theta_1, \ldots, \theta_K] ^ T \geq 0 $ and $v$ be the dual variables associated with the constrains $ \hat{ \bm x } \geq  0$ and $  \bm s ^ T \hat {\bm x} = 1 $, ${\bm x} $ should be the solution of the \gls{KKT} system \cite{boyd_vandenberghe_2004}:
        \begin{equation} \label{eq-kkt1}
            x_k = d_k + \theta_k - v s_k = 0,  \bm s^T \bm x = 1, \theta_k x_k = 0, x_k \geq 0, \theta_k \geq 0 ,
        \end{equation}
        where $x_k$ is the $k$-th element in $\bm x$, for $k = 1,\ldots,K$.
        
        To solve \eqref{eq-kkt1},	we let $r_k =  d_k / s_k$.
        Without loss of generality, it is assumed that $r_1 \leq \cdots \leq r_K$. 
        First, we note that $v$ should obey $v < r_K$.
        If $v \geq r_K$, we have $ x_k^ 2 = x_k \theta_k + x_k (d_k - v s_k) \leq 0 $
        for all $k$. Then $\bm x = 0$, which is not feasible.
        There should exist an integer $1 \leq j \leq K $ such that  $ v \in [ r_{j-1} ,r_j )$.
        If $j = 1$, the interval is $(- \infty, r_1)$.
        Similarly, we have $x_k = 0$ for all $ k < j$.
        In addition, for $k \geq  j$, we have $x_k - \theta_k =  d_k - v s_k   > 0$, so $\theta_k = 0$. 
        Based on these conclusions, we have
        \begin{equation} \label{eq-x}
            x_k = 0, \ \forall k < j, \  x_k = d_k - v s_k, \ \forall k \geq j.
        \end{equation}
        Since $\bm s^T \bm x = 1$, $v$ should be given by
        \begin{equation} \label{eq-v}
            v = \frac{\sum_{k=j}^{K} d_k - 1}{ \sum_{k=j}^{K} s_k  }.
        \end{equation}

        The problem is that $j$ is unknown. Nevertheless, we can search for a $j$ so that $  v(j) \in  [ r_{j-1} ,r_j ) $, where $v(j)$ is the value of $v$ computed by \eqref{eq-v}.
        Since $j$ is an integer in $[1,K]$, the times of searching is not larger than $K$.
        Once $j$ is obtained, the projection $\bm x$ can be computed via \eqref{eq-x} and \eqref{eq-v}.
        
        \section{The dual of \eqref{eq-balancing-dpc}}  \label{A-6} 
        
         The Lagrange function \cite{boyd_vandenberghe_2004} of \eqref{eq-minimal-power2} is
        \begin{equation}
        \mathcal{L}_2(  \bm F_u, \lambda, \bm Y, \bm d ) = \lambda + \mathrm{tr}\big(\bm Y   ( \bm F_u \bm F_u ^ H - \lambda  \bm I_r ) \big) + \sum_{k=1}^{K} d_k \Big\{ \sum_{i > k} \left|\bm u_k ^ H \bm f _ i\right|^2 + 1  -\frac{1}{\gamma} \left|\bm u_k ^ H \bm f _ k\right|^2 \Big\},
        \end{equation}
        where $\bm d  = \left[ d_1,\ldots,d_K \right]^T \geq 0$ and $\bm Y  \succeq 0 $ are dual variables.
        The dual function is  $	g_2(\bm Y, \bm d) = \min _ {\bm F_u,\lambda} \mathcal{L}_2 (\bm F_u,\lambda, \bm Y, \bm d )$. 
        It can be shown that $g_2(\bm Y, \bm d)$ is finite under the  conditions of \eqref{eq-dual-cond-a} and \eqref{eq-dual-cond-b}  \cite{4203115}, repeated as follows:
        \begin{subequations} \label{eq-dual-cond}
        	\begin{align}                       
        		\mathrm{tr} (\bm Y  ) & = 1, \tag{\ref{eq-dual-cond-a}} \\
        		\bm Y+  \sum_{i<k} d_i \bm u_i \bm u_i ^ H &\succeq \frac{1}{\gamma} d_k \bm u_k \bm u_k ^ H, \ k =1,\ldots,K. \tag{\ref{eq-dual-cond-b}}
        	\end{align}
        \end{subequations}
        When these conditions hold, the dual function is given by
        $ 	g_2(\bm Y, \bm d) =  \sum_{k=1}^{K} d_k$. 
        Therefore, the dual  problem of \eqref{eq-minimal-power2} is 
        \begin{equation} \label{eq-minimal-power2-dual}
        	\max_{\bm Y \succeq 0} \ \max_ {\bm d \geq 0}  \   \sum_{k=1}^{K} d_k , \ \
        	\mathrm{s.t.} \ \eqref{eq-dual-cond-a} \ \mathrm{and}  \ \eqref{eq-dual-cond-b}.
        \end{equation}

        It can be proven that  \eqref{eq-minimal-power2}   has strong duality \cite{4203115}, so the optimal value of \eqref{eq-minimal-power2-dual}, denoted by  $ \lambda_\mathrm{dual}^*(  \gamma)$,  equals to $\lambda^*(  \gamma)$. 
        Then, $\gamma_o^*$  should satisfy $\lambda_\mathrm{dual}^*(\gamma_o^* ) = 1$.
        In other words, $\gamma_o^*$ is  the minimal $\gamma \geq 0$ so that $ \lambda_\mathrm{dual}^*(\gamma ) \geq 1$.
        Note that $ \lambda_\mathrm{dual}^*(\gamma ) \geq 1$ means there exists a pair of feasible solution $\bm Y, \bm d$ in \eqref{eq-minimal-power2-dual} satisfying $ \sum_{k=1}^{K} d_k \geq 1$.  
        Therefore, $ \gamma_o^*$ is equal to the optimal value of the dual problem in \eqref{eq-sinr-balancing-dual}.
        
        \section{The gradient of the objective function in \eqref{eq-sinr-balancing-dual-2}}  \label{A-4} 
        
        To derive the gradient, we note that $\gamma_o (\bm Y),  \bm  d ^ * (\bm Y)  $ is the solution of the equations  in \eqref{eq-dd}.
        When $ \bm Y \succ 0$, the index set $ \mathcal{I}(\bm Y)$ is empty.
        Computing  the differential to the equations in \eqref{eq-dd}, we have
        \begin{subequations}   \label{eq-diff}
            \begin{align}	
                &	\bm 1_K ^ T \mathrm{d} \bm d = 0,  \label{eq-diff-a} \\
               & \mathrm{d}  \gamma - \mathrm{d}  d_k \bm u_k ^ H \bm f_k ^ * (\bm Y)   = - d_k  ^ * (\bm Y)  [ \bm f_k ^ * (\bm Y)] ^ H (  \mathrm{d}  \bm Y + \sum_{i < k}  \mathrm{d}  d_i \bm u_k \bm u_k ^ H  )  \bm   f_k ^ * (\bm Y)  ,  \label{eq-diff-b}
            \end{align}
        \end{subequations}
        for $k = 1,\ldots,K$. 
        We rewrite \eqref{eq-diff-b} into a matrix form:
        \begin{equation} \label{eq-diff-2}
            \bm A  \mathrm{d} \bm d =  \mathrm{d} \gamma \bm 1_K + \bm b,
        \end{equation}
        where the $k$-th element in  $ \bm b $ is  $ d_k ^ * (\bm Y) [ \bm f_k ^ * (\bm Y)] ^ H  \mathrm{d}  \bm Y \bm f_k ^ * (\bm Y)   $ and $\bm A$ is defined in \eqref{eq-A}.
        Combining \eqref{eq-diff-a} and \eqref{eq-diff-2}, one has  
        \begin{equation}
            \bm 1 _K ^ T \bm A ^ {-1} \bm 1_K \mathrm{d} \gamma + \bm 1_K ^ T \bm A ^ {-1} \bm b = 0.
        \end{equation}
        Letting $a_k$ be the $k$-th element in $\bm a = (\bm A ^ {-1}) ^ T \bm 1_K$, there is
        \begin{equation}
            \mathrm{d} \gamma =  - \frac{1}{\sum_{k=1}^{K}  a_k } \mathrm{tr} \Big\{  \sum_{k=1}^{K}  a_k d_k^*  (\bm Y) \hat{\bm f}_k^*(\bm Y) 
            [ \hat{\bm f}_k^*(\bm Y)] ^ H  \mathrm{d} \bm Y \Big\},
        \end{equation}
        from which one can obtain the gradient in \eqref{eq-gradY}.

        \section{Computation of the projection onto $\Omega_2$ } \label{A-3} 		
        The projection of a Hermitian matrix $\bm Y \in \mathbb{C} ^ {r \times r}  $ onto $\Omega_2$ is 
        \begin{equation} \label{eq-X}
            \arg \min_{{\bm X}} \frac{1}{2} \|  \bm Y - {\bm X}  \|_F ^ 2, \ \mathrm{s.t.}  \ \mathrm{tr} ({\bm X} ) = 1, \ { \bm  X} \succeq 0.
        \end{equation}
        To solve \eqref{eq-X}, we write the eigen decomposition of  $\bm Y$ as $ \bm Y  = \bm V \bm \Sigma_y \bm Y^H $, where $\bm V$ is a $r \times r$ unitary matrix and $\bm \Sigma_y$ is diagonal.
        Letting $\bm \Sigma_x = \bm V^ H {\bm X} \bm V $,  the optimization in \eqref{eq-X} is reformulated to
        \begin{equation} \label{eq-X1} 
            \min_{\bm \Sigma_x} \frac{1}{2} \|  \bm \Sigma_y - \bm \Sigma  _x \|_F ^ 2, \ \mathrm{s.t.}  \ \mathrm{tr} ({\bm \Sigma}_x ) = 1, \ {\bm \Sigma}_x \succeq 0.
        \end{equation}
        It can be observed that $ \bm \Sigma_x  $ should be diagonal at the optimum.
        We  let $y_k$ and $x_k$ be the $(k,k)$-th elements in $\bm \Sigma_y$ and $\bm \Sigma_x$, respectively. 
        Then \eqref{eq-X1} is equivalent to
        \begin{equation} \label{eq-x2}
            \min_{\bm x } \frac{1}{2} \|  \bm y - \bm x  \|_2  ^ 2, \ \mathrm{s.t.}  \ \bm 1_r ^ T \bm x = 1, \ \bm x \geq  0.
        \end{equation}
        where $\bm x = [x_1, \ldots, x_r] ^ T$ and $\bm y = [y_1, \ldots, y_r] ^ T$.
        Here,  \eqref{eq-x2} has the same form as the optimization in \eqref{eq-optx}, and can be solved with no more than $K$ loops.
        Once the optimal $\bm x$ is obtained, the projection is given by $ \bm V \mathrm{diag} (\bm x) \bm V^H $.

        \section{Proof for Theorem~\ref{theorem-1}} \label{A-1}
        In this proof, we verify that $\bm Y^*$ and $\bm d ^ *$ meet the \gls{KKT} condition of   \eqref{eq-weighted-rate-dual}, which is stated as
        \begin{subequations} \label{eq-kkt-yd}
            \begin{align}
                & 	( {\bm Y} ^*+  \sum_{k=1}^{K} d_k^* \bm u_ k \bm u_k ^ H ) ^{-1} -   [ \bm Y]^{*-1} + \nu \bm I_r  = \bm 0, \label{eq-kkt-yd-1} \\
                &  \mathrm{tr} ( {\bm Y ^ *})  = 1,\  {\bm Y}^* \succ 0, \label{eq-kkt-yd-3}  \\
                &  \bm u_j ^H ( {\bm Y}^*+  \sum_{k=1}^{K} d_k^* \bm u_k \bm u_k ^ H   ) ^{-1} \bm u _ j  - \eta+ \varphi_j = 0,  \label{eq-kkt-yd-2}  \\
                &  \sum_{k=1}^{K} d_k^* =1 ,\   d^*_j \geq 0,\ \varphi_j \geq 0, \ \varphi_j d_j^* = 0,  \label{eq-kkt-yd-4}
            \end{align}
        \end{subequations} 
        for $j = 1,\ldots,K$.
        Here, $\nu$, $\eta$ and $\{\varphi_k\}$ are the dual variables associated with the constraints $\mathrm{tr} (\bm Y) = 1$, $\sum_{k=1}^{K} d_k = 1$ and $\bm d \geq 0$, respectively, and their values are given by
        \begin{equation}
            \eta=  \bm 1_K ^ T \bm \phi,\ \nu =  \eta,\ \varphi_k =\eta\Big(1 - \bm u _ k ^ H  {\bm Z}^* \bm u _ k\Big),
        \end{equation}
        for $k = 1,\ldots,K$.
        
        In the following, we check the conditions in \eqref{eq-kkt-yd} one by one.
        First, we point out an important relationship between $\bm d^*$, $ {\bm Y}^*$ and $ {\bm Z}^*$.
        From  \eqref{eq-Zphi}, one has
        \begin{equation}
            (1 / \eta) [ {\bm Z} ^ {*}]^{-1} =  \bm Y ^ * + \sum_{k=1}^{K}  d_k^* \bm u_k  \bm u_k ^ H.
        \end{equation}
        Then 
        \begin{displaymath} 
            \begin{aligned}
                &  ( {\bm Y}^*+  \sum_{k=1}^{K} d_k^* \bm u_k \bm u_k ^ H ) ^{-1} -  [ \bm Y^{*}]^{-1} =  \eta  {\bm Z}^* -  [ \bm Y^{*}]^{-1}  = - \nu \bm I_r,
            \end{aligned}       
        \end{displaymath}
        i.e. \eqref{eq-kkt-yd-1} holds, and 
        \begin{displaymath}
            \bm u_j^H ( {\bm Y}^*+  \sum_{ki=1}^{K} d_k^* \bm u_k \bm u_k ^ H   ) ^{-1} \bm u _ j = \eta \bm u_j ^ H \bm Z^* \bm u _j =  \eta - \varphi_k,
        \end{displaymath}
        i.e. \eqref{eq-kkt-yd-2}  holds.
        In \eqref{eq-kkt-yd-3}, $ {\bm Y}^* \succ 0$ is trivial, and 
        \begin{equation} \label{eq-trys}
            \mathrm{tr} ( {\bm Y}^*) = \frac{1}{\eta}	\mathrm{tr} ( (\bm Z^* + \bm I_r) ^ {-1} )  =\frac{r}{\eta} - \frac{1}{\eta}  \mathrm{tr} (  (\bm Z^* + \bm I_r) ^ {-1}  \bm Z^* ). 
        \end{equation}
        Multiply the  left and right side of \eqref{eq-Zphi} by $ {\bm Z} ^ *$ and take the matrix trace, we have
        \begin{equation} \label{eq-trace-eq}
            \mathrm{tr} (  (\bm Z^* + \bm I_r) ^ {-1}  \bm Z^* ) = r - \sum_{k=1}^{K} \phi_k \bm u_k^H \bm Z^* \bm u_k.
        \end{equation}
        Substituting \eqref{eq-trace-eq} into \eqref{eq-trys}, one has
        \begin{equation}
            \mathrm{tr} ( \bm Y ^* ) = \frac{1}{\eta} \sum_{k=1}^{K} \phi_k  \bm u_k^H \bm Z^* \bm u_k = \frac{1}{\eta} \sum_{k=1}^{K} \phi_k   = 1, 
        \end{equation}
        so \eqref{eq-kkt-yd-3} holds. 
        Finally we show that \eqref{eq-kkt-yd-4} holds.
        The first three conditions in \eqref{eq-kkt-yd-4} are trivial.
        According to  \eqref{eq-Zphi}, it can be shown that
        \begin{eqnarray}
            \varphi_j d_j ^ * =  \phi_j  (1 -  \bm u _ j^ H  {\bm Z}^* \bm u_j ) = 0.
        \end{eqnarray} 
        Therefore \eqref{eq-kkt-yd-4} holds and the proof is completed.
        
    \end{appendices} 
    
    \bibliographystyle{IEEEtran}
    \bibliography{ref}

\end{document}